\documentclass[pra,floatfix,showpacs,twocolumn,nofootinbib]{revtex4-1}

\usepackage{times,amsmath,amsfonts,amssymb,latexsym,amsthm,mathtools,natbib,multirow,epsf,latexsym,setspace,array}
\usepackage{graphicx,graphics,epsfig,epsf}
\usepackage[normalem]{ulem}
\usepackage{bm,,bbm,mathrsfs,color,xcolor}
\usepackage{hypernat,hyperref}
\usepackage{dsfont}
\usepackage{tikz}
\usetikzlibrary{matrix,arrows,shapes,backgrounds,patterns,calc}
\usepackage{pgfplots}
\usepackage{soul} 
\usetikzlibrary{external}
\newtheorem{theorem}{Theorem}

\newtheorem{observation}[theorem]{Observation}

\newtheorem{mylemma}{Lemma}
\newtheorem{mydefinition}{Definition}
\newtheorem{mytheorem}{Theorem}

\newtheorem{myexample}{Example}

\newcommand{\tikzimageswap}[2]{#1} 

\newcommand{\ket}[1]{\left| #1 \right\rangle}
\newcommand{\bra}[1]{\left \langle #1 \right |}
\newcommand{\braket}[2]{\langle #1 | #2 \rangle}
\newcommand{\ketbra}[2]{| #1 \rangle\! \langle #2 |}

\newcommand{\ignore}[1]{}

\DeclareMathOperator{\Tr}{Tr}
\DeclareMathOperator{\Ad}{Ad}
\DeclareMathOperator{\Span}{Span}
\DeclareMathOperator{\id}{id}
\newcommand{\TrB}{\Tr_{\textsc{b}}}
\newcommand{\identity}{\mathds{1}}
\newcommand{\Sgp}{\mathcal{G}} 
\newcommand{\HH}{\mathcal{H}}  
\newcommand{\HS}{\mathcal{H}_{\textsc{s}}}
\newcommand{\HB}{\mathcal{H}_{\textsc{b}}}
\newcommand{\HW}{\mathcal{H}_{\textsc{w}}}
\newcommand{\HSB}{\HS\otimes\HB}
\newcommand{\HSBW}{\HSB\otimes\HW}
\newcommand{\BL}{\mathcal{B}}
\newcommand{\BH}{\BL(\HH)}
\newcommand{\BHS}{\BL(\HS)}
\newcommand{\BHB}{\BL(\HB)}
\newcommand{\BHW}{\BL(\HW)}
\newcommand{\BHSB}{\BL(\HSB)}
\newcommand{\BHSW}{\BL(\HS\otimes\HW)}
\newcommand{\BHSBW}{\BL(\HSBW)}

\newcommand{\US}{\mathrm{U}(\HS)}
\newcommand{\UB}{\mathrm{U}(\HB)}
\newcommand{\UW}{\mathrm{U}(\HW)}
\newcommand{\USB}{\mathrm{U}(\HSB)}
\newcommand{\DD}[1]{\mathcal{D}_{\textsc{#1}}}
\newcommand{\DS}{\DD{s}}
\newcommand{\DB}{\DD{b}}
\newcommand{\DW}{\DD{w}}
\newcommand{\DSB}{\DD{sb}}
\newcommand{\DSW}{\DD{sw}}

\newcommand{\DSBW}{\DD{sbw}}
\newcommand{\rhoS}{\rho_{\textsc{s}}}
\newcommand{\rhoB}{\rho_{\textsc{b}}}
\newcommand{\rhoSB}{\rho_{\textsc{sb}}}

\newcommand{\CS}{\mathcal{V}} 
\newcommand{\AS}{{\mathcal{S}}} 
\newcommand{\AM}{\mathcal{A}} 

\newcommand{\jmdstack}[2]{\genfrac{}{}{0pt}{}{#1}{#2}}

\begin{document}

\title{A general framework for complete positivity}
\author{Jason M. Dominy$^{(2,4)}$, Alireza Shabani$^{(5,6)}$, and Daniel A. Lidar$^{(1,2,3,4)}$}
\affiliation{Departments of $^{(1)}$Chemistry, $^{(2)}$Electrical Engineering, and $^{(3)} $Physics, $^{(4)}$Center for Quantum Information Science \& Technology, University of Southern California, Los Angeles, CA 90089, USA\\
$^{(5)}$Department of Chemistry, University of California, Berkeley, CA 94720, USA\\
$^{(6)}$Google Research, Venice, CA, 90291, USA}

\begin{abstract}
Complete positivity of quantum dynamics is often viewed as a litmus test for physicality, yet it is well known that correlated initial states need not give rise to completely positive evolutions. This observation spurred numerous investigations over the past two decades attempting to identify necessary and sufficient conditions for complete positivity. Here we describe a complete and consistent mathematical framework for the discussion and analysis of complete positivity for correlated initial states of open quantum systems.  This formalism is built upon a few simple axioms and is sufficiently general to contain all prior methodologies going back to Pechakas, PRL (1994) \cite{Pechukas:1994}.  The key observation is that initial system-bath states with the same reduced state on the system must evolve under all admissible unitary operators  to system-bath states with the same reduced state on the system, in order to ensure that the induced dynamical maps on the system are well-defined. Once this consistency condition is imposed, related concepts like the assignment map and the dynamical maps are uniquely defined.  In general, the dynamical maps may not be applied to arbitrary system states, but only to those in an appropriately defined physical domain. We show that the constrained nature of the problem gives rise to not one but three inequivalent types of complete positivity. Using this framework we elucidate the limitations of recent attempts to provide conditions for complete positivity using quantum discord and the quantum data-processing inequality. 
The problem remains open, and may require fresh perspectives and new mathematical tools.  The formalism presented herein may be one step in that direction.
\end{abstract}
\maketitle

\newcommand{\spherePlane}[3]{
	\def\r{1.2} 
	\def\plht{1.5} 
	\def\plwd{2} 
		
	\filldraw[shift={(0:#1)}, shift={(90:#2)}, fill = \plcolor, fill opacity = 0.5] (-\plwd,-\plht) rectangle (\plwd,\plht);
	
	\filldraw[shift={(0:#1)}, shift={(90:#2)}, pattern = north west lines] (0,0) circle [radius = #3];

	\filldraw[ball color=\spcolor, fill opacity = 0.4] (#1,#2) circle[radius=\r]; 
}	

\newcommand{\rotatedSpherePlaneRight}[4]{
	\def\r{1.2} 
	\def\plht{1.5} 
	\def\plwd{2} 
		
	\fill[shift={(0:#1)}, shift={(90:#2)}, domain=0:360, smooth cycle, variable = \g, pattern = north west lines, opacity = 0.5, fill = \plcolor] plot ({cameraProjectX(\r*cos(#3)*cos(\g), \r*sin(#3)*cos(\g), \r*sin(\g))},{cameraProjectY(\r*cos(#3)*cos(\g), \r*sin(#3)*cos(\g), \r*sin(\g))});

	\filldraw[shift={(0:#1)}, shift={(90:#2)}, domain=0:360, smooth cycle, variable = \g, pattern = north west lines, fill opacity = 0.5] plot ({cameraProjectX(#4*cos(#3)*cos(\g), #4*sin(#3)*cos(\g), #4*sin(\g))},{cameraProjectY(#4*cos(#3)*cos(\g), #4*sin(#3)*cos(\g), #4*sin(\g))});
	
	\filldraw[color = \plcolor, fill opacity = 0.5, draw opacity = 0, shift={(0:#1)}, shift={(90:#2)}, domain=(-180-asin((\r*\cz*sin(#3) - \cy*sqrt(\cy*\cy+(\cz*\cz - \r*\r)*sin(#3)*sin(#3)))/(\cy*\cy+\cz*\cz*sin(#3)*sin(#3)))):(180-asin((\r*\cz*sin(#3) + \cy*sqrt(\cy*\cy+(\cz*\cz - \r*\r)*sin(#3)*sin(#3)))/(\cy*\cy+\cz*\cz*sin(#3)*sin(#3)))), smooth, variable = \g] ($cameraProjectX(\r*cos(#3)*cos(180-asin((\r*\cz*sin(#3) + \cy*sqrt(\cy*\cy+(\cz*\cz - \r*\r)*sin(#3)*sin(#3)))/(\cy*\cy+\cz*\cz*sin(#3)*sin(#3)))), \r*sin(#3)*cos(180-asin((\r*\cz*sin(#3) + \cy*sqrt(\cy*\cy+(\cz*\cz - \r*\r)*sin(#3)*sin(#3)))/(\cy*\cy+\cz*\cz*sin(#3)*sin(#3)))), \plht)*(1,0) + cameraProjectY(\r*cos(#3)*cos(180-asin((\r*\cz*sin(#3) + \cy*sqrt(\cy*\cy+(\cz*\cz - \r*\r)*sin(#3)*sin(#3)))/(\cy*\cy+\cz*\cz*sin(#3)*sin(#3)))), \r*sin(#3)*cos(180-asin((\r*\cz*sin(#3) + \cy*sqrt(\cy*\cy+(\cz*\cz - \r*\r)*sin(#3)*sin(#3)))/(\cy*\cy+\cz*\cz*sin(#3)*sin(#3)))), \plht)*(0,1)$) -- ($cameraProjectX(\plwd*cos(#3), \plwd*sin(#3), \plht)*(1,0) + cameraProjectY(\plwd*cos(#3), \plwd*sin(#3), \plht)*(0,1)$) -- ($cameraProjectX(\plwd*cos(#3), \plwd*sin(#3), -\plht)*(1,0) + cameraProjectY(\plwd*cos(#3), \plwd*sin(#3), -\plht)*(0,1)$) -- ($cameraProjectX(\r*cos(#3)*cos(-180-asin((\r*\cz*sin(#3) - \cy*sqrt(\cy*\cy+(\cz*\cz - \r*\r)*sin(#3)*sin(#3)))/(\cy*\cy+\cz*\cz*sin(#3)*sin(#3)))), \r*sin(#3)*cos(-180-asin((\r*\cz*sin(#3) - \cy*sqrt(\cy*\cy+(\cz*\cz - \r*\r)*sin(#3)*sin(#3)))/(\cy*\cy+\cz*\cz*sin(#3)*sin(#3)))), -\plht)*(1,0) + cameraProjectY(\r*cos(#3)*cos(-180-asin((\r*\cz*sin(#3) - \cy*sqrt(\cy*\cy+(\cz*\cz - \r*\r)*sin(#3)*sin(#3)))/(\cy*\cy+\cz*\cz*sin(#3)*sin(#3)))), \r*sin(#3)*cos(-180-asin((\r*\cz*sin(#3) - \cy*sqrt(\cy*\cy+(\cz*\cz - \r*\r)*sin(#3)*sin(#3)))/(\cy*\cy+\cz*\cz*sin(#3)*sin(#3)))), -\plht)*(0,1)$) -- plot ({cameraProjectX(\r*cos(#3)*cos(\g), \r*sin(#3)*cos(\g), \r*sin(\g))},{cameraProjectY(\r*cos(#3)*cos(\g), \r*sin(#3)*cos(\g), \r*sin(\g))});
	\draw[shift={(0:#1)}, shift={(90:#2)}] ($cameraProjectX(\r*cos(#3)*cos(180-asin((\r*\cz*sin(#3) + \cy*sqrt(\cy*\cy+(\cz*\cz - \r*\r)*sin(#3)*sin(#3)))/(\cy*\cy+\cz*\cz*sin(#3)*sin(#3)))), \r*sin(#3)*cos(180-asin((\r*\cz*sin(#3) + \cy*sqrt(\cy*\cy+(\cz*\cz - \r*\r)*sin(#3)*sin(#3)))/(\cy*\cy+\cz*\cz*sin(#3)*sin(#3)))), \plht)*(1,0) + cameraProjectY(\r*cos(#3)*cos(180-asin((\r*\cz*sin(#3) + \cy*sqrt(\cy*\cy+(\cz*\cz - \r*\r)*sin(#3)*sin(#3)))/(\cy*\cy+\cz*\cz*sin(#3)*sin(#3)))), \r*sin(#3)*cos(180-asin((\r*\cz*sin(#3) + \cy*sqrt(\cy*\cy+(\cz*\cz - \r*\r)*sin(#3)*sin(#3)))/(\cy*\cy+\cz*\cz*sin(#3)*sin(#3)))), \plht)*(0,1)$) -- ($cameraProjectX(\plwd*cos(#3), \plwd*sin(#3), \plht)*(1,0) + cameraProjectY(\plwd*cos(#3), \plwd*sin(#3), \plht)*(0,1)$) -- ($cameraProjectX(\plwd*cos(#3), \plwd*sin(#3), -\plht)*(1,0) + cameraProjectY(\plwd*cos(#3), \plwd*sin(#3), -\plht)*(0,1)$) -- ($cameraProjectX(\r*cos(#3)*cos(-180-asin((\r*\cz*sin(#3) - \cy*sqrt(\cy*\cy+(\cz*\cz - \r*\r)*sin(#3)*sin(#3)))/(\cy*\cy+\cz*\cz*sin(#3)*sin(#3)))), \r*sin(#3)*cos(-180-asin((\r*\cz*sin(#3) - \cy*sqrt(\cy*\cy+(\cz*\cz - \r*\r)*sin(#3)*sin(#3)))/(\cy*\cy+\cz*\cz*sin(#3)*sin(#3)))), -\plht)*(1,0) + cameraProjectY(\r*cos(#3)*cos(-180-asin((\r*\cz*sin(#3) - \cy*sqrt(\cy*\cy+(\cz*\cz - \r*\r)*sin(#3)*sin(#3)))/(\cy*\cy+\cz*\cz*sin(#3)*sin(#3)))), \r*sin(#3)*cos(-180-asin((\r*\cz*sin(#3) - \cy*sqrt(\cy*\cy+(\cz*\cz - \r*\r)*sin(#3)*sin(#3)))/(\cy*\cy+\cz*\cz*sin(#3)*sin(#3)))), -\plht)*(0,1)$);

	\draw[shift={(0:#1)}, shift={(90:#2)}, domain=(-180-asin((\r*\cz*sin(#3) - \cy*sqrt(\cy*\cy+(\cz*\cz - \r*\r)*sin(#3)*sin(#3)))/(\cy*\cy+\cz*\cz*sin(#3)*sin(#3)))):(180-asin((\r*\cz*sin(#3) + \cy*sqrt(\cy*\cy+(\cz*\cz - \r*\r)*sin(#3)*sin(#3)))/(\cy*\cy+\cz*\cz*sin(#3)*sin(#3)))), smooth, variable = \g] plot ({cameraProjectX(\r*cos(#3)*cos(\g), \r*sin(#3)*cos(\g), \r*sin(\g))},{cameraProjectY(\r*cos(#3)*cos(\g), \r*sin(#3)*cos(\g), \r*sin(\g))});
	
	\filldraw[ball color=\spcolor, fill opacity = 0.4] (#1,#2) circle[radius=\r*sqrt((\cy*\cy + \cz*\cz)/(\cy*\cy+\cz*\cz-\r*\r))];
	
	\filldraw[color = \plcolor, fill opacity = 0.5, draw opacity = 0, shift={(0:#1)}, shift={(90:#2)}, domain=(180-asin((\r*\cz*sin(#3) + \cy*sqrt(\cy*\cy+(\cz*\cz - \r*\r)*sin(#3)*sin(#3)))/(\cy*\cy+\cz*\cz*sin(#3)*sin(#3)))):(180-asin((\r*\cz*sin(#3) - \cy*sqrt(\cy*\cy+(\cz*\cz - \r*\r)*sin(#3)*sin(#3)))/(\cy*\cy+\cz*\cz*sin(#3)*sin(#3)))), smooth, variable = \g] ($cameraProjectX(\r*cos(#3)*cos(180-asin((\r*\cz*sin(#3) + \cy*sqrt(\cy*\cy+(\cz*\cz - \r*\r)*sin(#3)*sin(#3)))/(\cy*\cy+\cz*\cz*sin(#3)*sin(#3)))), \r*sin(#3)*cos(180-asin((\r*\cz*sin(#3) + \cy*sqrt(\cy*\cy+(\cz*\cz - \r*\r)*sin(#3)*sin(#3)))/(\cy*\cy+\cz*\cz*sin(#3)*sin(#3)))), -\plht)*(1,0) + cameraProjectY(\r*cos(#3)*cos(180-asin((\r*\cz*sin(#3) + \cy*sqrt(\cy*\cy+(\cz*\cz - \r*\r)*sin(#3)*sin(#3)))/(\cy*\cy+\cz*\cz*sin(#3)*sin(#3)))), \r*sin(#3)*cos(180-asin((\r*\cz*sin(#3) + \cy*sqrt(\cy*\cy+(\cz*\cz - \r*\r)*sin(#3)*sin(#3)))/(\cy*\cy+\cz*\cz*sin(#3)*sin(#3)))), -\plht)*(0,1)$) -- ($cameraProjectX(-\plwd*cos(#3), -\plwd*sin(#3), -\plht)*(1,0) + cameraProjectY(-\plwd*cos(#3), -\plwd*sin(#3), -\plht)*(0,1)$) -- ($cameraProjectX(-\plwd*cos(#3), -\plwd*sin(#3), \plht)*(1,0) + cameraProjectY(-\plwd*cos(#3), -\plwd*sin(#3), \plht)*(0,1)$) -- ($cameraProjectX(\r*cos(#3)*cos(180-asin((\r*\cz*sin(#3) - \cy*sqrt(\cy*\cy+(\cz*\cz - \r*\r)*sin(#3)*sin(#3)))/(\cy*\cy+\cz*\cz*sin(#3)*sin(#3)))), \r*sin(#3)*cos(180-asin((\r*\cz*sin(#3) - \cy*sqrt(\cy*\cy+(\cz*\cz - \r*\r)*sin(#3)*sin(#3)))/(\cy*\cy+\cz*\cz*sin(#3)*sin(#3)))), \plht)*(1,0) + cameraProjectY(\r*cos(#3)*cos(180-asin((\r*\cz*sin(#3) - \cy*sqrt(\cy*\cy+(\cz*\cz - \r*\r)*sin(#3)*sin(#3)))/(\cy*\cy+\cz*\cz*sin(#3)*sin(#3)))), \r*sin(#3)*cos(180-asin((\r*\cz*sin(#3) - \cy*sqrt(\cy*\cy+(\cz*\cz - \r*\r)*sin(#3)*sin(#3)))/(\cy*\cy+\cz*\cz*sin(#3)*sin(#3)))), \plht)*(0,1)$) -- plot ({cameraProjectX(\r*cos(#3)*cos(\g), \r*sin(#3)*cos(\g), \r*sin(\g))},{cameraProjectY(\r*cos(#3)*cos(\g), \r*sin(#3)*cos(\g), \r*sin(\g))});
	\draw[shift={(0:#1)}, shift={(90:#2)}] ($cameraProjectX(\r*cos(#3)*cos(180-asin((\r*\cz*sin(#3) + \cy*sqrt(\cy*\cy+(\cz*\cz - \r*\r)*sin(#3)*sin(#3)))/(\cy*\cy+\cz*\cz*sin(#3)*sin(#3)))), \r*sin(#3)*cos(180-asin((\r*\cz*sin(#3) + \cy*sqrt(\cy*\cy+(\cz*\cz - \r*\r)*sin(#3)*sin(#3)))/(\cy*\cy+\cz*\cz*sin(#3)*sin(#3)))), -\plht)*(1,0) + cameraProjectY(\r*cos(#3)*cos(180-asin((\r*\cz*sin(#3) + \cy*sqrt(\cy*\cy+(\cz*\cz - \r*\r)*sin(#3)*sin(#3)))/(\cy*\cy+\cz*\cz*sin(#3)*sin(#3)))), \r*sin(#3)*cos(180-asin((\r*\cz*sin(#3) + \cy*sqrt(\cy*\cy+(\cz*\cz - \r*\r)*sin(#3)*sin(#3)))/(\cy*\cy+\cz*\cz*sin(#3)*sin(#3)))), -\plht)*(0,1)$) -- ($cameraProjectX(-\plwd*cos(#3), -\plwd*sin(#3), -\plht)*(1,0) + cameraProjectY(-\plwd*cos(#3), -\plwd*sin(#3), -\plht)*(0,1)$) -- ($cameraProjectX(-\plwd*cos(#3), -\plwd*sin(#3), \plht)*(1,0) + cameraProjectY(-\plwd*cos(#3), -\plwd*sin(#3), \plht)*(0,1)$) -- ($cameraProjectX(\r*cos(#3)*cos(180-asin((\r*\cz*sin(#3) - \cy*sqrt(\cy*\cy+(\cz*\cz - \r*\r)*sin(#3)*sin(#3)))/(\cy*\cy+\cz*\cz*sin(#3)*sin(#3)))), \r*sin(#3)*cos(180-asin((\r*\cz*sin(#3) - \cy*sqrt(\cy*\cy+(\cz*\cz - \r*\r)*sin(#3)*sin(#3)))/(\cy*\cy+\cz*\cz*sin(#3)*sin(#3)))), \plht)*(1,0) + cameraProjectY(\r*cos(#3)*cos(180-asin((\r*\cz*sin(#3) - \cy*sqrt(\cy*\cy+(\cz*\cz - \r*\r)*sin(#3)*sin(#3)))/(\cy*\cy+\cz*\cz*sin(#3)*sin(#3)))), \r*sin(#3)*cos(180-asin((\r*\cz*sin(#3) - \cy*sqrt(\cy*\cy+(\cz*\cz - \r*\r)*sin(#3)*sin(#3)))/(\cy*\cy+\cz*\cz*sin(#3)*sin(#3)))), \plht)*(0,1)$);
	
	\draw[shift={(0:#1)}, shift={(90:#2)}, domain=(180-asin((\r*\cz*sin(#3) + \cy*sqrt(\cy*\cy+(\cz*\cz - \r*\r)*sin(#3)*sin(#3)))/(\cy*\cy+\cz*\cz*sin(#3)*sin(#3)))):(180-asin((\r*\cz*sin(#3) - \cy*sqrt(\cy*\cy+(\cz*\cz - \r*\r)*sin(#3)*sin(#3)))/(\cy*\cy+\cz*\cz*sin(#3)*sin(#3)))), smooth, variable = \g] plot ({cameraProjectX(\r*cos(#3)*cos(\g), \r*sin(#3)*cos(\g), \r*sin(\g))},{cameraProjectY(\r*cos(#3)*cos(\g), \r*sin(#3)*cos(\g), \r*sin(\g))});
}

\section{Introduction}
Completely positive (CP) maps have played an important role in the long and extensive history of the problem of the formulation and characterization of the dynamics of open quantum systems \cite{Kraus:book,Breuer:book}. They have become a widespread tool, e.g., in quantum information science \cite{Nielsen:book}. Thus it is of interest to establish under which conditions CP maps can arise from a complete description of an open system that includes both the system and its environment or bath. To arrive at a map one first identifies an admissible set of initial system-bath states which is one-to-one with the corresponding set of system states via the partial trace.  Then one jointly and unitarily evolves the system and bath, and then traces out the bath. It is well known that if the set of admissible initial system-bath states is completely uncorrelated (product states) with a fixed bath state then this ``standard procedure" gives rise to a CP map description of the evolution of the system. The situation is far more complicated when the set of initial system-bath states contains correlated states.
The basic reason for this is that when correlations are present in the initial system-bath state, a clean separation between system and bath is no longer possible, and the ``standard procedure" alone no longer uniquely defines a map between the set of all initial and final system states.

Subsystem dynamical maps are uniquely defined by the joint unitary evolution of system and bath and the choice of the set of initial states and several studies have pointed out that entangled initial system-bath states can lead to non-CP maps \cite{Pechukas:1994,Pechukas:1995, Jordan:2004,Carteret:2008}. Rodriguez-Rosario \textit{et al.} highlighted the role of quantum discord \cite{Ollivier:2002} and showed that a CP map arises if the set of joint initial system-bath states is purely classically correlated, i.e., has vanishing quantum discord \cite{Rod:08}. Subsequently Shabani \& Lidar showed that, under certain additional constraints, the vanishing quantum discord condition is not just sufficient but also necessary for complete positivity \cite{SL,SL2}. {As will be discussed in Section \ref{sec:ShabaniLidar2009}, those additional constraints were not recognized at the time and weaken the conclusions of \cite{SL}.  More recently, Brodutch \emph{et al.} \cite{Brodutch:2013} and Buscemi \cite{Buscemi:2014} demonstrated that this connection between complete positivity and discord does not generalize to all cases.  Brodutch \emph{et al.} did so by offering a counterexample in the form of a set of initial system-bath states, almost all of which are discordant, which nevertheless exhibit completely positive subdynamics.  Buscemi followed this up by describing a general method for constructing examples yielding completely positive subdynamics, even though they may feature highly entangled states.  It remains an open problem, therefore, to fully elucidate the relationship between structural features of the set of initial system-bath states and the behavior of the resulting dynamics, including whether or not the dynamics are completely positive.  

We do not offer a complete solution to the problem in these pages.  Rather, it is our purpose here to establish a complete and consistent mathematical framework for the discussion and analysis of linear subsystem dynamics, including the question of complete positivity for correlated initial states.  Additionally, we develop the idea that, in the case of correlated initial states, there are several inequivalent definitions of complete positivity.  Physical, mathematical, and operational concerns may recommend one flavor of complete positivity over the others, as we discuss. Our analysis builds on the notion of what we call ``$\Sgp$-consistent operator spaces", that represent the set of admissible initial system-bath states.  $\Sgp$-consistency is necessary to ensure that initial system-bath states with the same reduced state on the system evolve under all admissible unitary operators to system-bath states with the same reduced state on the system, ensuring that the dynamical maps on the system are well-defined.

The paper is organized as follows.  In Section \ref{sec:subsysDynMaps}, we examine the minimal conditions necessary for the existence of well-defined subsystem dynamical maps, and formally define $\Sgp$-consistent subsets.  The more specific case of linear dynamical maps is considered in Section \ref{sec:GconsistLinSp}, where we formally define $\Sgp$-consistent subspaces, consider properties of the uniquely defined ``assignment maps'' associated with these spaces, describe the subsystem dynamical maps related to a $\Sgp$-consistent subspace, and operator sum representations of these maps.  Section \ref{sec:CPity} develops some definitions of complete positivity, discusses the physical motivation of these ideas, and extends these definitions to apply to ``assignment maps''.  We present in Section \ref{sec:Examples} a number of examples culled from the existing literature on completely positive dynamics and show how they may be expressed in the present framework.   Finally, a summary and discussion of open questions is offered in Section \ref{sec:Summary}.

\section{Subsystem Dynamical Maps}
\label{sec:subsysDynMaps}
The time-evolution of a quantum subsystem is often described in terms of dynamical maps that transform the reduced state of the system at time $t_{0}$ to the reduced state at time $t_{1}$.  However, the future state of a quantum system in contact with its environment depends upon not just the current state of the system, but the joint state of the system and environment, as well as the joint evolution.  For this reason, such dynamical maps are not generally well-defined: each possible reduced state of the system may arise from infinitely many distinct system-bath states, which after joint evolution, can yield infinitely many distinct possible reduced system states at a future time $t_{1}$.  We wish to describe a mathematical framework which is broad enough to contain the various dynamical map constructions of \cite{SL,SL2,Rod:08,Brodutch:2013, Buscemi:2014} and others.  It is closely related to the assignment map approach first described by Pechukas \cite{Pechukas:1994}, but is more general and places greater emphasis on the space of admissible initial system-bath states.  In order to build this framework involving dynamical maps for subsystem evolution, it is necessary to remove the indeterminacy of the final system state from the problem, requiring that we begin by making an assumption about the set of admissible initial system-bath states, which we describe presently.

To set the stage, let us briefly review the setting of open quantum systems.  For any Hilbert space $\HH$, let $\BL(\HH)$ denote the algebra of all bounded linear operators on $\HH$ and let $\mathrm{U}(\HH)$ denote the unitary group on $\HH$.  For any $U\in\mathrm{U}(\HH)$, the adjoint map $\Ad_{U}\in\BL(\BL(\HH))$ is the conjugation superoperator $A\mapsto UAU^{\dag}$.   Consider a system $S$ with a $d_{\textsc{s}}$-dimensional system Hilbert-space $\mathcal{H}_{\textsc{s}}$, a bath $B$ with a $d_{\textsc{b}}$-dimensional Hilbert space $\mathcal{H}_{\textsc{b}}$ and the joint system-bath Hilbert space $\mathcal{H}_{\textsc{s}}\otimes \mathcal{H}_{\textsc{b}}$.  Let $\DS\subset\BHS$, $\DB\subset\BHB$, and $\DSB\subset\BHSB$ denote the convex sets of density matrices on $\HS$, $\HB$, and $\HSB$, respectively.  The joint system-bath state $\rho _{\textsc{sb}}(t) \in \DSB$ evolves under a joint time-dependent unitary propagator $U(t)$ as $\Ad_{U(t)}[\rhoSB(0)]=\rho _{\textsc{sb}}(t)$.  The reduced system and bath states are given via the partial trace operation by $\rhoS(t)=\TrB[\rhoSB(t)] \in \DS$ and $\rhoB(t)=\Tr_{\textsc{s}}[\rhoSB(t)] \in \DB$, respectively. The standard prescription for the system subdynamics is thus given by the following quantum dynamical process (QDP), acting on the initial system-bath state $\rhoSB(0)$:
\begin{equation}
	\rhoS(t)=\TrB[U(t)\rhoSB(0)U^{\dag }(t)].
	\label{eq:rho_S(t)}
\end{equation}
It is important to note that, in contrast to some authors (e.g., \cite{Alicki:1995}), the map we seek is one that describes the time-evolution of the post-preparation state of the system.  Evolution from an idealized state through experimental realization (i.e., state preparation) ending with unitary evolution of the system-bath state should be modeled with an additional ``preparation'' map precomposed with the evolution map described herein.

Consider some designated set of admissible initial states $\AS\subset\DSB$.  We shall later assume that $\AS$ is convex, but for the time being we do not restrict its generality.
\begin{mydefinition}[$\AS$-dynamical map]
Given $U\in\USB$, a subsystem $\AS$-dynamical map $\tau_{U}^{\AS}:\TrB\AS\to\DS$ is any map that satisfies $\tau_{U}^{\AS}(\TrB\rho) = \TrB(U\rho U^{\dag})$ for all $\rho\in\AS$, i.e., that makes the diagram in Figure \ref{fig:admissibleStates} commute.  
\end{mydefinition}
\begin{figure}
	\tikzimageswap{\includegraphics{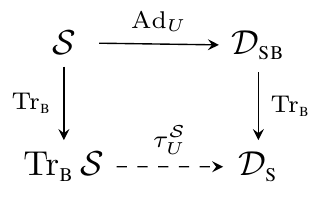}}{
	\tikzsetnextfilename{figure1}
	\begin{tikzpicture}
		\matrix (m) [matrix of math nodes, ampersand replacement=\&, row sep=2em,column sep=3em,minimum width=2em]
		{
			\AS \& \DSB \\
	  		\TrB\AS \& \DS\\ 
	  	};
	  	\path[-stealth,font=\scriptsize]
		(m-1-1) edge node [above] {$\Ad_{U}$} (m-1-2)
		(m-1-1) edge node [left] {$\TrB$} (m-2-1)
		(m-2-1) edge [dashed] node [above] {$\tau_{U}^{\AS}$} (m-2-2)
		(m-1-2) edge node [right] {$\TrB$} (m-2-2);		
	\end{tikzpicture}}
	\caption{The set of admissible initial system-bath states $\AS\subset\DSB$ must be chosen such that a well-defined subsystem dynamical map $\tau_{U}^{\AS}:\TrB\AS\to\DS$ exists which makes the diagram commute, i.e., which satisfies $\tau_{U}^{\AS}(\TrB \rho) = \TrB(U\rho U^{\dag})$ for all $\rho\in\AS$.  This requirement defines the concept of a $U$-consistent subset \cite{diagram-comment}.}
	\label{fig:admissibleStates}
\end{figure}
\noindent Such a map can only be well-defined if, whenever $\rho_{1},\rho_{2}\in\AS$ and $\TrB\rho_{1} = \TrB\rho_{2}$, it follows that $\TrB(U\rho_{1} U^{\dag}) =\TrB(U\rho_{2} U^{\dag})$.  This defines a necessary property of $\AS$ that we call \emph{$U$-consistency}.

We typically want to define not just a single dynamical map, but a family of them based on a subsemigroup of unitary system-bath evolution operators (``semi" since we will set $t\geq 0$). 
This gives rise to the notion of $\Sgp$-consistency.
\begin{mydefinition}[$\Sgp$-Consistent Subset]
	Let $\Sgp\subset\USB$ be a subsemigroup of the unitary group acting on the Hilbert space $\HSB$.  A subset of system-bath states $\AS\subset\DSB$ will be called $\Sgp$-consistent if it is $U$-consistent for all $U\in \Sgp$, i.e., if, whenever $\rho_{1},\rho_{2}\in\AS$ are such that $\TrB\rho_{1} = \TrB\rho_{2}$, $\TrB(U\rho_{1}U^{\dag}) = \TrB(U\rho_{2}U^{\dag})$ for all $U\in \Sgp$.
\end{mydefinition}
\noindent 
$\Sgp$ represents the semigroup of ``allowed'' unitary evolutions of system and bath, in other words, the set of unitary evolutions, closed under compositions, for which the subsystem dynamical maps are of interest.  For example, if the system and bath will only evolve according to a single fixed time-independent Hamiltonian $H$, then $\Sgp$ is typically the one-parameter subsemigroup $\Sgp = \{e^{-itH}\;:\; t\geq 0\}$.  If the system and/or bath are subject to control, then $\Sgp$ may be a larger subsemigroup representing all unitary operators that may be generated within the control scheme.  Generally, the larger the semigroup $\Sgp$, the more restrictions $\Sgp$-consistency places on the subset of admissible initial states $\AS$ and the smaller such a subset must be.  Indeed, when $\Sgp = \USB$, the condition becomes that $\TrB$ must describe a one-to-one correspondence between $\AS$ and $\TrB\AS$, i.e., for each system state $\rhoS\in\DS$, there exists \emph{at most}\footnote{Not every state in $D_{\textsc{s}}$ needs to be covered by a state in $\AS$.  Those that are not covered are inadmissible initial system states, lying outside the domain of $\tau_{U}^{\AS}$.  See also the comments at the end of Section \ref{sec:constrainedDomains}.} one state $\rhoSB\in\AS$ such that $\TrB\rhoSB = \rhoS$ \cite{G-consistency-comment}. 
\begin{myexample}
\label{ex:DiffBathStates}
As a simple example that not all subsets $\AS\subset\DSB$ are $\USB$-consistent, consider a system and bath, each one qubit, and states $\rhoS\in\DS$ and $\rhoB^{(1)}\neq\rhoB^{(2)}\in\DB$.  Let $U\in\USB$ be the swap operator $U(\ket{\psi}\otimes\ket{\phi}) = \ket{\phi}\otimes\ket{\psi}$.  Then the set $\AS := \{\rhoSB^{(1)} = \rhoS\otimes\rhoB^{(1)}, \rhoSB^{(2)} = \rhoS\otimes\rhoB^{(2)}\}$ is not $U$-consistent since $\TrB\rhoSB^{(1)} = \TrB\rhoSB^{(2)} = \rhoS$, but $\TrB(U\rhoSB^{(1)}U^{\dag}) = \rhoB^{(1)}\neq\rhoB^{(2)} = \TrB(U\rhoSB^{(2)}U^{\dag})$.
\end{myexample}

\begin{myexample}
\label{ex:StelBu}
\v{S}telmachovi\v{c} and Bu\v{z}ek \cite{Stelmachovic:2001} offered another example which  demonstrates that not all subsets $\AS\subset\DSB$ are $\USB$-consistent, i.e., that there exist initial states $\rho_{1},\rho_{2}\in\DSB$ and $U\in\USB$ such that $\TrB\rho_{1} = \TrB\rho_{2}$, but $\TrB(U\rho_{1}U^{\dag})\neq\TrB(U\rho_{2}U^{\dag})$.  Specifically, they considered a system and bath comprising one qubit each, such that
\begin{subequations}
\begin{align}
	\rho_{1} & = |\alpha|^{2}\ketbra{00}{00} + |\beta|^{2}\ketbra{11}{11}\\
	\rho_{2} & = (\alpha\ket{00} + \beta\ket{11})(\alpha^{*}\bra{00} + \beta^{*}\bra{11})\\
	U & = -i(\ketbra{1}{1}\otimes\sigma_{x} + \ketbra{0}{0}\otimes \identity) \simeq \textsc{cnot},
\end{align}
\end{subequations}
which satisfy 
\begin{subequations}
\begin{align}
	\TrB\rho_{1} & = |\alpha|^{2}\ketbra{0}{0} + |\beta|^{2}\ketbra{1}{1} = \TrB\rho_{2}\\
	\TrB(U\rho_{1}U^{\dag}) & = |\alpha|^{2}\ketbra{0}{0} + |\beta|^{2}\ketbra{1}{1}\\
	\TrB(U\rho_{2}U^{\dag}) & = (\alpha\ket{0} + \beta\ket{1})(\alpha^{*}\bra{0} + \beta^{*}\bra{1}),
\end{align}
\end{subequations}
so that $\TrB(U\rho_{1}U^{\dag})\neq \TrB(U\rho_{2}U^{\dag})$.  It follows that if $\textsc{cnot}\in\Sgp$, then no $\Sgp$-consistent subset $\AS\subset\DSB$ may contain both $\rho_{1}$ and $\rho_{2}$.
\end{myexample}

The specification of a set $\AS\subset\DSB$ of admissible initial states may be thought of as a ``promise'' that the initial system-bath state will always lie in $\AS$.  The constraint that $\AS$ must be $\Sgp$-consistent can impose heavy restrictions on this set, often forcing $\AS$ to be very small relative to $\DSB$. It will be instructive to consider the consequences of constraining to a $\Sgp$-consistent subset $\AS$ in order to understand the properties of the resulting dynamical map and to place earlier studies in their proper context.

The subsystem maps $\tau_{U}^{\AS}$ defined as in this section may exhibit a wide variety of behavior (including non-linearity) that depends on the choice of $\AS$.  It may be noted however that all maps $\tau_{U}^{\AS}$ defined in this way (as well as the linear subsystem dynamical maps $\Psi_{U}^{\CS}$ defined in Section \ref{sec:GconsistLinSp}) share the property of being trace-preserving.  That is a simple consequence of the fact that the system-bath evolution is assumed to the trace-preserving (due to unitarity) and that the partial trace is trace-preserving.  Issues associated with population loss therefore never arise in this formalism.

\section{Linear Dynamical Maps}
\label{sec:GconsistLinSp}

\subsection{\texorpdfstring{$\Sgp$}{G}-consistent linear subspaces}
Properties of the set of admissible initial states $\AS$ are closely related to properties of the resulting dynamical maps $\tau_{U}^{\AS}$, and it is reasonable to wonder if there are either necessary or quite common properties of subsystem dynamics that should be incorporated into the framework we are developing.  One such property, nearly ubiquitous in the open quantum system literature (see \cite{Alicki:1995} for a notable exception we discuss in Section~\ref{sec:Alicki}), is that of convex-linearity of $\tau_{U}^{\AS}$, i.e., the property that for any $\rhoS,\sigma_{\textsc{s}}\in\TrB\AS$, and any $\alpha\in[0,1]$, 
\begin{equation}
	\tau_{U}^{\AS}(\alpha\rhoS + (1-\alpha)\sigma_{\textsc{s}}) = \alpha\tau_{U}^{\AS}(\rhoS) + (1-\alpha)\tau_{U}^{\AS}(\sigma_{\textsc{s}}).
\end{equation}
Clearly, this statement about $\tau_{U}^{\AS}$ requires that $\TrB\AS$ be a convex set.  We will go a bit further and assume that $\AS$ itself is a convex subset of $\BHSB$.  Any subset $\AS\subset \DSB$ may be uniquely extended to the $\mathbb{C}$-linear subspace 
\begin{equation}
	\CS = \Span_{\mathbb{C}}\AS\subset\BHSB
\end{equation}
and the maps $\TrB\big|_{\AS}$, and $\Ad_{U}\big|_{\AS}$ may be likewise uniquely extended by linearity to $\TrB\big|_{\CS}$, and $\Ad_{U}\big|_{\CS}$.  By assuming that $\AS$ is convex and $\Sgp$-consistent, we may similarly extend the maps $\tau_{U}^{\AS}:\TrB\AS\to\DS$ to $\mathbb{C}$-linear maps $\Psi_{U}^{\CS}:\TrB\CS\to\BHS$ for all $U\in\Sgp$.  It should be emphasized that, while $\tau_{U}^{\AS}$ is defined only on states, the map $\Psi_{U}^{\CS}$ is defined on the linear operator space $\TrB\CS\subset\BHS$ which contains both states, and non-state operators.  

It is readily verified that the subspace $\CS$ constructed as the complex span of a convex, $\Sgp$-consistent $\AS$ exhibits the following properties:
\begin{enumerate}
	\item $\CS\subset\BHSB$ is a $\mathbb{C}$-linear subspace.
	\item $\CS$ is spanned by states, i.e., $\Span_{\mathbb{C}}(\DSB\cap\CS) = \CS$.\label{item:positiveSpan} 
	\item $\CS$ is $\Sgp$-consistent, i.e., if $X,Y\in\CS$ are such that $\TrB X = \TrB Y$ and if $U\in\Sgp$, then $\TrB(UXU^{\dag}) =\TrB(UYU^{\dag})$. \label{item:GconsLinSp}
\end{enumerate} 
Any $\mathbb{C}$-linear subspace $\CS\subset\BHSB$ will be called a \emph{$\Sgp$-consistent subspace} if it satisfies these three properties. Figure~\ref{fig:G-consistent} illustrates how a $\Sgp$-consistent subspace relates to $\DSB$.

Note that it is implied by $\mathbb{C}$-linearity of $\CS$ and property \ref{item:positiveSpan} that $\CS$ is self-adjoint, i.e., $X\in\CS$ implies $X^{\dag}\in\CS$. Note further that any $\mathbb{C}$-linear subspace $\CS\subset\BHSB$ which is self-adjoint and spanned by states is always a $\Sgp$-consistent subspace when $\Sgp$ is the trivial group $\Sgp=\{\identity\}$.  

\begin{figure}
	\scalebox{0.8}{
	\tikzimageswap{\includegraphics{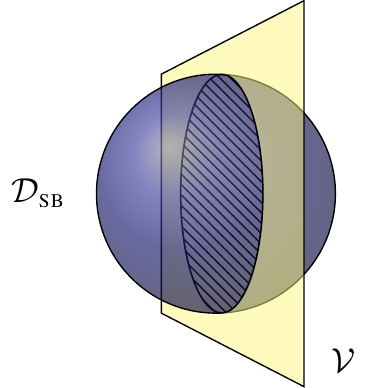}}{
	\def\cy{8} 
	\def\cz{0}
	\def\plcolor{yellow!50}
	\def\spcolor{blue!80}
	\tikzsetnextfilename{figure2}
	\begin{tikzpicture}[declare function={
			cameraProjectX(\xx,\yy,\zz)=\xx*(\cy*\cy + \cz*\cz)/(\cy*(\cy-\yy) + \cz*(\cz-\zz)); 
			cameraProjectY(\xx,\yy,\zz)=(\cy*\zz-\cz*\yy)*(sqrt(\cy*\cy + \cz*\cz))/(\cy*(\cy-\yy) + \cz*(\cz-\zz));
		}]

		\rotatedSpherePlaneRight{1.5}{6.0}{-110}{1.2}
		\draw (2.8,4.3) node {$\CS$};
		\draw (-0.3,6.0) node {$\DSB$};
		
	\end{tikzpicture}
	}}
	\caption{Schematic representation of a $\Sgp$-consistent subspace $\CS$ and its relationship to $\DSB$, the convex set of all system-bath density matrices, represented by the ball.  The intersection $\CS\cap\DSB$, indicated by the ruled area, is such that $\CS = \Span_{\mathbb{C}}(\CS\cap\DSB)$.}
	\label{fig:G-consistent}
\end{figure}

Within any $\mathbb{C}$-linear subspace $\CS\subset\BHSB$, we may identify a further subspace $\CS_{0}\subset\CS$ by 
	\begin{equation}
		\CS_{0}:= \ker\big(\TrB\big|_{\CS}\big) = \{X\in\CS \;:\; \TrB(X) = 0\}.
	\end{equation}
Then property \ref{item:GconsLinSp}, $\Sgp$-consistency, is equivalent to the property that $\Sgp\cdot\CS_{0}\subset \ker \TrB$, i.e., for any $X\in\CS_{0}$ and $U\in \Sgp$, $\TrB(UXU^{\dag})=0$.  
In other words, differences in states with the same reduced state cannot play any role in the final reduced state after evolution: $\TrB(\rho-\sigma) = 0$ for $\rho,\sigma\in\CS\cap\DSB$ must imply $\TrB[U(\rho-\sigma)U^{\dag}] = 0$ for all $U\in\Sgp$.
The kernel of this idea goes back at least to \cite{SL2} (section II.C), and possibly earlier.

\subsection{Is dynamical map composition meaningful?}
\label{sec:dynSemigroups}
It should be noted that, given some system-bath Hilbert space $\HSB$, $U,U'\in\USB$, and $U$-consistent subspace $\CS\subset\BHSB$ and $U'$-consistent subspace $\CS'\subset\BHSB$, we cannot in general compose the dynamical subsystem maps $\Psi_{U}^{\CS}$ and $\Psi_{U'}^{\CS'}$ to form $\Psi_{U'}^{\CS'}\circ\Psi_{U}^{\CS}$.  The reason is that $\Ad_{U}\CS$ need not lie in $\CS'$, so that the image of $\Psi_{U}^{\CS}$ need not lie in the domain of $\Psi_{U'}^{\CS'}$.  As a result, the composition $\Psi_{U'}^{\CS'}\circ\Psi_{U}^{\CS}$ may be meaningless.  For example, for a given $\Sgp$-consistent subspace $\CS$, the dynamical maps $\{\Psi_{U}^{\CS}\;:\;U\in\Sgp\}$ \emph{do not} comprise a semigroup in general.  Care must be taken when analysing sequences of these dynamical maps that the compositions are meaningful.

\subsection{Assignment maps}
Given a $\Sgp$-consistent subspace $\CS$, consider the quotient space $\CS/\CS_{0}$ in which each element corresponds to an affine subspace of the form $X+\CS_{0}\subset \CS$ for some $X\in\CS$.   Because of linearity, property \ref{item:GconsLinSp} above is equivalent to the statement that if $Y,Z\in\CS$ are such that $\TrB Y = \TrB Z$, then $\TrB(UYU^{\dag}) = \TrB(UZU^{\dag})$ for all $U\in \Sgp$.   It follows that the linear maps $\TrB:\CS\to\TrB\CS$ and $\Ad_{U}:\CS\to\BHSB$ for each $U\in \Sgp$ uniquely define linear maps $\TrB:\CS/\CS_{0}\to\TrB\CS$ and $\Ad_{U}:\CS/\CS_{0}\to\BHSB/\ker\TrB$.  Moreover, the map $\TrB:\CS/\CS_{0}\to\TrB\CS$ is a one-to-one correspondence, and so may be inverted, yielding the ``assignment map''.

\begin{mydefinition}[Assignment Map]
The assignment map $\AM_{\CS}:\TrB\CS\to\CS/\CS_{0}$ is defined by
\begin{equation}
	\AM_{\CS}(X) = \big(\TrB\big|_{\CS}\big)^{-1}(X) = \{Y\in\CS\;:\;\TrB Y = X\}.
\label{eq:AM-def}
\end{equation}
for any $X\in\TrB\CS$.  For each $X\in\TrB\CS$, the output $\AM_{\CS}(X)$ may be thought of either as an affine subspace of $\CS$, or as a point in the vector space $\CS/\CS_{0}$.
\end{mydefinition}
This definition of the assignment map is a generalization of the original definition introduced by Pechukas \cite{Pechukas:1994}.  A Pechukas-like assignment map (which assigns to each operator in $\TrB\CS$ a unique operator in $\CS$, rather than an affine subspace) is recovered if and only if the $\Sgp$-consistent subspace $\CS$ is $\USB$-consistent, as demonstrated in the following lemma.
\begin{mylemma}
	\label{lem:USBconsistency}
	Let $\CS\subset\BHSB$ be a $\mathbb{C}$-linear, self-adjoint subspace which is spanned by states, i.e., a $\{\identity\}$-consistent subspace.  $\CS$ is $\USB$-consistent if and only if $\TrB\big|_{\CS}$ is injective, i.e., $\CS_{0}:=\ker\big(\TrB\big|_{\CS}\big) = \{0\}$.  In this case, $\CS/\CS_{0} = \CS$, so that the assignment map carries each operator in $\TrB\CS$ to a unique operator in $\CS$, as in Pechukas' original definition of the assignment map \cite{Pechukas:1994}.
\end{mylemma}
\begin{proof}
	First, observe that if $\CS_{0} = \{0\}$, then $\Sgp\cdot\CS_{0} = \CS_{0}\subset\ker\TrB$ for any subsemigroup $\Sgp\subset\USB$.  Therefore $\CS$ is $\Sgp$-consistent for any $\Sgp\subset\USB$, and, in particular, is $\USB$-consistent.  On the other hand, if $\CS$ is $\USB$-consistent, then $\Span_{\mathbb{C}}\big[\USB\cdot\CS_{0}\big]\subset\ker\TrB$ is an invariant subspace of the adjoint representation of $\USB$ on $\mathrm{sl}(\HS\otimes\HB) = \ker\Tr$ [the zero trace operators in $\BHSB$].  Since $\ker\TrB$ is a proper subspace of $\ker\Tr$ (assuming $\dim\HS>1$), this invariant subspace generated by $\CS_{0}$ cannot be all of $\mathrm{sl}(\HS\otimes\HB)$.  The ``only if'' part of the lemma then follows from the irreducibility  of the adjoint representation of $\USB$ on $\mathrm{sl}(\HS\otimes\HB)$ \cite[\textsection 2.4.4]{Goodman:2009}), which implies that the invariant subspace must be $\{0\}$, and therefore $\CS_{0} = \{0\}$. 
\end{proof}

\subsection{How constrained are the domains of dynamical maps?}
\label{sec:constrainedDomains}

Some authors have expressed a desire to keep $\TrB\CS = \BHS$, so that the domain of the dynamical maps $\Psi_{U}^{\CS}$ is unconstrained. Lemma \ref{lem:USBconsistency} shows that, when $\Sgp$ is all of $\USB$, any $\Sgp$-consistent subspace must be severely constrained to have dimension no higher than $\BHS$.  In Lemma \ref{lem:localUconsistency} we demonstrate that, if $\Sgp$ contains non-local (entangling) operators, then any $\Sgp$-consistent subspace $\CS$ must necessarily be a proper subspace of $\BHSB$, and therefore is constrained in some way.  
Specifically, Lemma \ref{lem:localUconsistency} says that if $\Sgp$ contains non-local unitaries, and $\CS = \BHSB$, then $\CS_{0} = \ker\TrB$ is such that $\Sgp\cdot\CS_{0}\not\subset \ker\TrB$, so that $\CS$ is not $\Sgp$-consistent.  Thus, the only $\Sgp$-consistent subspaces must be proper subspaces of $\BHSB$.

\begin{mylemma}
	\label{lem:localUconsistency}
	The full kernel 
	\begin{equation}
		\ker\TrB := \{X\in\BHSB\;:\; \TrB X = 0\}
	\end{equation}
	is $\Sgp$-invariant if and only if $\Sgp$ is a subsemigroup of the group $\US\otimes\UB$ of local unitary operators.  
\end{mylemma}
(Essentially this same question, posed and answered in different language, was the subject of \cite{Hayashi:2003}.)
\begin{proof}
	First, observe that ``if" is trivial, i.e., if $U\in\US\otimes\UB$, then $\Ad_{U}\ker\TrB \subset \ker\TrB$ since, if $U=U_{\textsc{s}}\otimes U_{\textsc{b}}$, then $\TrB(UXU^{\dag}) = U_{\textsc{s}}[\TrB X]U_{\textsc{s}}^{\dag} = 0$ for any $X\in\ker\TrB$.  
	
	To prove ``only if", i.e., if $\Ad_{U}\ker\TrB \subset \ker\TrB$, then $U\in\US\otimes\UB$, note first that $\BHS\otimes\identity$ is the orthogonal complement to $\ker\TrB$ in the Hilbert-Schmidt geometry.  This may be seen from the fact that $\BHSB = \BHS\otimes\identity \oplus \ker\TrB$ since any $X\in\BHSB$ may be uniquely decomposed as $X = \TrB(X)\otimes\identity/d_{\textsc{b}} + X_{\ker}$, where $X_{\ker} = X - \TrB(X)\otimes\identity/d_{\textsc{b}}\in \ker\TrB$, and $\BHS\otimes\identity$ and $\ker\TrB$ are orthogonal subspaces [if $A\in\BHS$ and $X\in\ker\TrB$, then $\langle A\otimes\identity, X\rangle_{\mathrm{HS}} = \langle A,\TrB X\rangle_{\mathrm{HS}} = 0$].  Now, $U\in\USB$ satisfies $\Ad_{U}\ker\TrB \subset \ker\TrB$ if and only if $0 = \langle \Ad_{U}X, A\otimes\identity\rangle_{\mathrm{HS}} = \langle X, \Ad_{U^{\dag}}(A\otimes\identity)\rangle_{\mathrm{HS}}$ for all $X\in\ker\TrB$ and $A\in\BHS$, i.e., if and only if $\Ad_{U^{\dag}}[\BHS\otimes\identity]\subset\BHS\otimes\identity$.  In other words, the condition is that, for any $A\in\BHS$, there exists $B\in\BHS$ such that $U^{\dag}(A\otimes\identity)U = B\otimes\identity$.  This condition implies, in particular, that $U$ belongs to the normalizer of $\US\otimes\identity$, which is $\US\otimes\UB$ \cite{Grace:2010}.
\end{proof}

It follows from Lemma \ref{lem:localUconsistency} that any $\mathbb{C}$-linear subspace $\CS\subset\BHSB$ which is spanned by states (i.e., a $\{\identity\}$-consistent subspace) is a $\US\otimes\UB$-consistent subspace.  This is because, for any such $\CS$, $\CS_{0} =\ker\big(\TrB|_{\CS}\big)\subset\ker\TrB$ is mapped by $\US\otimes\UB$ into $\ker\TrB$ (by Lemma \ref{lem:localUconsistency}), which is the condition for $\US\otimes\UB$-consistency.  In particular, for any $\Sgp\subset\US\otimes\UB$, $\CS = \BHSB$ is a valid $\Sgp$-consistent subspace and always gives rise to completely positive dynamics \cite{Salgado:2002}, since for any $U=U_{\textsc{s}}\otimes U_{\textsc{b}}\in\Sgp$,  $\Psi_{U}^{\CS}(X) = U_{\textsc{s}}X U_{\textsc{s}}^{\dag}$ for any $X\in\TrB\CS$, which is trivially completely positive.

With $\Sgp$-consistency imposing such strong constraints on the space of admissible system-bath operators, the desire to keep the domain of the dynamical maps unconstrained seems misplaced.  If it is reasonable to promise that the initial system-bath state will never lie outside a low-dimensional subspace $\CS$, it should be reasonable to also promise that the initial system state will never lie outside a proper subspace $\TrB\CS\subset\BHS$. Building on this observation we now proceed to identify a unique linear dynamical map for the subsystem.

\subsection{Linear dynamical maps from \texorpdfstring{$\Sgp$}{G}-consistent subspaces}
\label{sec:linDynMaps}
Note that the quotient space $\CS/\CS_{0}$ admits some additional structure that will be useful in characterizing the assignment map.  First, observe that $\CS_{0}$ is a self-adjoint $\mathbb{C}$-linear subspace.  An affine subspace $X+\CS_{0}\in\CS/\CS_{0}$ contains a Hermitian operator if $X + \CS_{0} = Y + \CS_{0}$ with $Y = Y^{\dag}$, which holds if and only if the affine subspace is self-adjoint, i.e., $(X+\CS_{0})^{\dag} = X+\CS_{0}$ 
Such an affine subspace will then be spanned by Hermitian operators, i.e., if $(X+\CS_{0})^{\dag} = X+\CS_{0}$ and $H_{X}$ is the set of Hermitian elements in $X + \CS_{0}$, then $X+\CS_{0}$ is the $\mathbb{C}$-affine hull of $H_{X}$.  We may also identify a closed convex cone $(\CS/\CS_{0})^{+}$ of ``positive'' elements in $\CS/\CS_{0}$ as the collection of affine subspaces in $\CS/\CS_{0}$ that each contain at least one positive operator.  Thus $\CS/\CS_{0}$ is an ordered vector space with a conjugate-linear involution $\dag$.  The partial trace $\TrB$ maps each affine subspace in $(\CS/\CS_{0})^{+}$ to a positive operator in $\TrB\CS$.  Likewise, we call the assignment map $\AM_{\CS}:\TrB\CS\to\CS/\CS_{0}$ a ``positive'' map if $\AM_{\CS}$ maps every positive operator in $\TrB\CS$ to an element of $(\CS/\CS_{0})^{+}$.

\begin{mylemma}
	The assignment map $\AM_{\CS}:\TrB\CS\to\CS/\CS_{0}$ defined in Eq.~\eqref{eq:AM-def} is $\mathbb{C}$-linear,  $\dag$-linear [i.e.,  $\AM_{\CS}(X^{\dag}) = \AM_{\CS}(X)^{\dag}$], and trace-preserving.  
\end{mylemma}
\begin{proof}
 	$\mathbb{C}$-linearity of $\AM_{\CS}$ follows from the linearity of $\TrB\big|_{\CS}$.  Similarly, for any $X\in\BHSB$, $\TrB(X+\CS_{0}) = \TrB(X)$ and $(X+\CS_{0})^{\dag} = X^{\dag} + \CS_{0}$, so that $\TrB[(X+\CS_{0})^{\dag}] = \TrB(X^{\dag}) = \TrB(X)^{\dag}$.  And if $\TrB(X+\CS_{0}) = \TrB(Y+\CS_{0})$, then $X-Y\in\CS_{0}$, so $X+\CS_{0} = Y+\CS_{0}$.  Therefore $\AM_{\CS}(\TrB(X)^{\dag})$ can only be $X^{\dag} + \CS_{0} = \AM_{\CS}(\TrB(X))^{\dag}$, so that $\AM_{\CS}$ is $\dag$-linear.  Finally, $\Tr(\AM_{\CS}(X)) = \Tr(\TrB[\AM_{\CS}(X)]) = \Tr(X)$ for any $X\in\TrB\CS$, so that $\AM_{\CS}$ is trace-preserving.
\end{proof}

We are now able to define a unique $\mathbb{C}$-linear dynamical map for the subsystem, presented in the following lemma.
\begin{mylemma}
	\label{lem:linHermMap}
	Let $\CS\subset \BHSB$ be a $\Sgp$-consistent subspace.  For any $U\in \Sgp$ there is a unique map $\Psi_{U}^{\CS}:\TrB\CS\to\BHS$ such that the diagram in Figure \ref{fig:commutativeDiagram2} commutes, i.e., such that, for any operator $X\in\CS$, $\Psi_{U}^{\CS}(\TrB X) = \TrB(UXU^{\dag})$.  That map $\Psi_{U}^{\CS}:\TrB\CS\to\BHS$, given by $\Psi_{U}^{\CS} = \TrB\circ\Ad_{U}\circ\AM_{\CS}$ is $\mathbb{C}$-linear, $\dag$-linear, and trace-preserving over the domain $\TrB\CS$, and acts as the dynamical map for system states in $\TrB(\DSB\cap\CS)\subset\DS$. 
\begin{figure}
	\tikzimageswap{\includegraphics{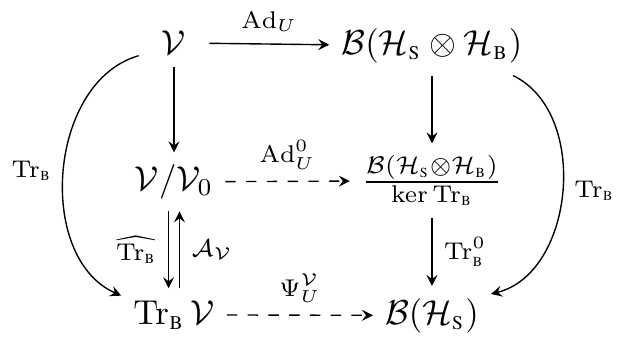}}{
	\tikzsetnextfilename{figure3}
	\begin{tikzpicture}
		\matrix (m) [matrix of math nodes, ampersand replacement=\&,row sep=2em,column sep=3em,minimum width=2em]
		{
			\CS \& \BHSB \\
	  		\CS/\CS_{0} \& \frac{\BHSB}{\ker\TrB} \\
	  		\TrB\CS \& \BHS\\ 
	  	};
	  	\path[-stealth,font=\scriptsize]
		(m-1-1) edge node [left] {} (m-2-1)
		(m-1-2) edge node [left] {} (m-2-2)
		(m-1-1) edge node [above] {$\Ad_{U}$} (m-1-2)
		([xshift=-0.35ex]m-2-1.south) edge node [left] {$\widehat{\TrB}$} ([xshift=-0.35ex]m-3-1.north)
		(m-2-1) edge [dashed] node [above] {$\Ad_{U}^{0}$} (m-2-2)
		(m-3-1) edge [dashed] node [above] {$\Psi_{U}^{\CS}$} (m-3-2)
		([xshift=0.35ex]m-3-1.north) edge node [right] {$\AM_{\CS}$} ([xshift=0.35ex]m-2-1.south)
		(m-2-2) edge node [right] {$\TrB^{0}$} (m-3-2)
		(m-1-1) edge [bend right=70] node [left] {$\TrB$} (m-3-1)
		(m-1-2) edge [bend left=70] node [right] {$\TrB$} (m-3-2);		
	\end{tikzpicture}}
	\caption{For any $\Sgp$-consistent subspace $\CS$ and unitary evolution operator $U\in\Sgp\subset\USB$, this commutative diagram uniquely defines the $\mathbb{C}$-linear, $\dag$-linear, trace-preserving map $\Psi_{U}^{\CS} = \TrB^{0}\circ\Ad_{U}^{0}\circ\AM_{\CS}$ which acts as the time evolution operator for system states in $\TrB(\DSB\cap\CS)\subset\DS$.  $\Sgp$-consistency is the condition that, for every $U\in\Sgp$, a map $\Ad_{U}^{0}:\CS/\CS_{0}\to\BHSB/\ker\TrB$ exists which makes the top rectangle commute.}
	\label{fig:commutativeDiagram2}
\end{figure}
\end{mylemma}  
\begin{proof}
	The uniqueness of $\Psi_{U}^{\CS}$ is due to the uniqueness of the assignment map $\AM_{\CS}$ and the maps $\Ad_{U}^{0}:\CS/\CS_{0}\to\BHSB/{\ker\TrB}$ and $\TrB^{0}:\BHSB/\ker\TrB\to\BHS$.  Since these maps are all $\mathbb{C}$-linear and $\dag$-linear, $\Psi_{U}$ is as well.  The $\dag$-linearity of $\Psi_{U}^{\CS}$ is equivalently expressed as $\Psi_{U}$ being Hermiticity-preserving, which is sometimes shortened to just ``Hermitian''. For any state $\rhoS\in\TrB(\DSB\cap\CS)$, {the affine subspace $\AM_{\CS}(\rhoS)$} must contain a state in $\DSB$, so the transformation $\rhoS\mapsto\Psi_{U}^{\CS}(\rhoS)$ reflects the unitary evolution of a valid system-bath state in $\CS$, and therefore $\Psi_{U}^{\CS}$ is the dynamical map for such a state.
\end{proof}

\begin{mydefinition}[Physical Domain]
	We call the convex set $\TrB(\DSB\cap\CS)\subset\DS$ the \emph{physical domain} because, as described in Lemma \ref{lem:linHermMap}, these are the system states for which the maps $\Psi_{U}^{\CS}$ act as the physical dynamical maps.  This is called the ``compatibility domain'' in \cite{Jordan:2004}.  The ``promise'' that initial system-bath states will lie in $\CS\cap\DSB$ implies a ``promise'' that initial system-states will lie in the physical domain.
\end{mydefinition}

We stress a few key points concerning this construction:
\begin{enumerate}
\item The linearity of the maps $\AM_{\CS}$ and $\Psi_{U}^{\CS}$ is due to our choice to assume that the set of admissible initial system-bath states belong to a linear subspace.  Depending on the physical processes involved, other choices are possible, leading to non-linear assignment maps and non-linear evolution operators $\Psi_{U}^{\CS}$ \cite{Romero:2004, Stelmachovic:2001, Breuer:book}.  
\item The map $\Psi_{U}^{\CS}$ does not generally act as the evolution operator for all system states, or even for all system states in $\DS\cap\TrB\CS$.  Any state $\rho\in\DS\cap\TrB\CS$ which does not lie in the convex ``physical domain'' $\TrB(\DSB\cap\CS)$ is mapped by the assignment map to an affine subspace of $\BHSB$ which contains no valid system-bath states in $\DSB$.  Since the transformation by $\Psi_{U}^{\CS}$ of such a system state is not tied to the unitary evolution of a valid system-bath state, it is empty of any physical meaning.  Such a map $\Psi_{U}^{\CS}$ should never be described without clearly indicating the physical domain of system states on which it can meaningfully be applied.  To apply the map outside this domain is to overinterpret the mathematics. 

\item For each $U\in \Sgp$, the definition of $\Psi_{U}^{\CS}$ depends entirely upon the $\Sgp$-consistent subspace $\CS$.  Different consistent subspaces, even if they include some shared fiducial state $\rhoSB$, will yield different maps $\Psi_{U}^{\CS}$.
\end{enumerate}

\begin{figure}
	\scalebox{0.8}{
	\tikzimageswap{\includegraphics{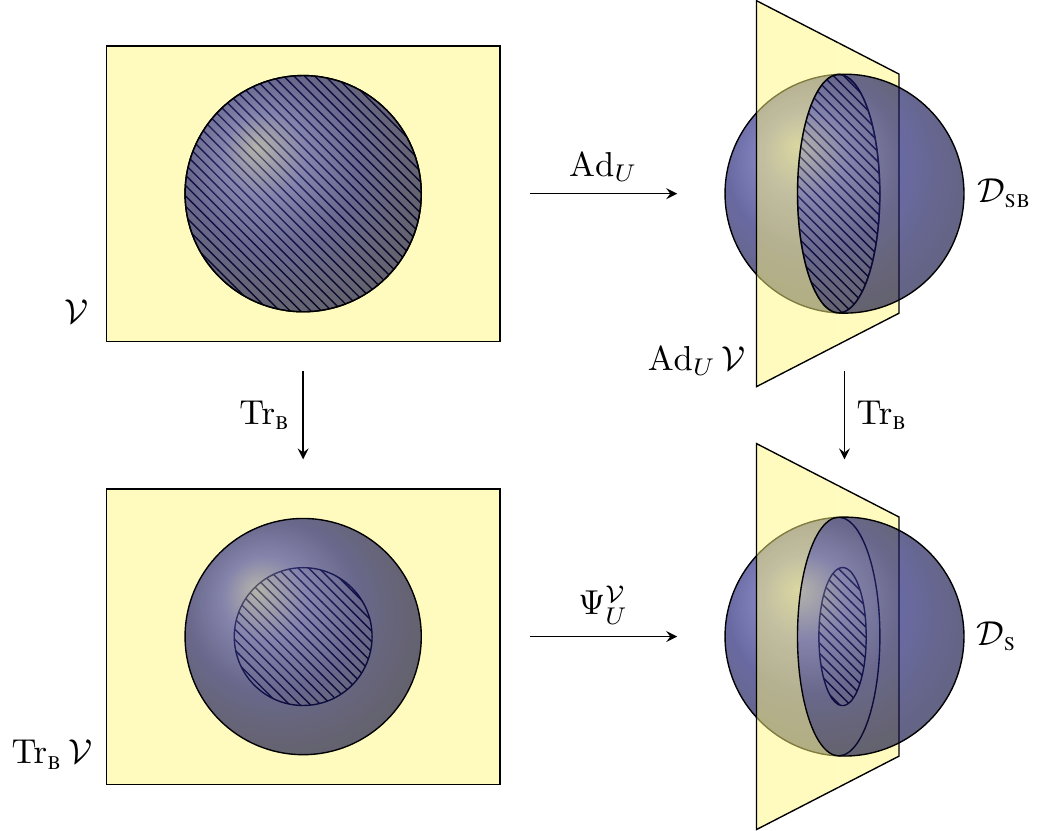}}{
	\def\cy{8} 
	\def\cz{0}
	\def\plcolor{yellow!50}
	\def\spcolor{blue!80}
	\tikzsetnextfilename{figure4}
	\begin{tikzpicture}[declare function={
			cameraProjectX(\xx,\yy,\zz)=\xx*(\cy*\cy + \cz*\cz)/(\cy*(\cy-\yy) + \cz*(\cz-\zz)); 
			cameraProjectY(\xx,\yy,\zz)=(\cy*\zz-\cz*\yy)*(sqrt(\cy*\cy + \cz*\cz))/(\cy*(\cy-\yy) + \cz*(\cz-\zz));
		}]
		\spherePlane{2}{6.0}{1.2}
		\draw (-.3,4.8) node {$\CS$};
		
		\draw[->,>=stealth] (2,4.2) to node [left] {$\TrB$} (2,3.3);
		
		\spherePlane{2}{1.5}{0.7}
		\draw (-.55,0.3) node {$\TrB\CS$};

		\draw[->,>=stealth] (4.3,6.0) to node [above] {$\Ad_{U}$} (5.8,6.0);

		\rotatedSpherePlaneRight{7.5}{6.0}{-70}{1.2}
		\draw (6.0,4.3) node {$\Ad_{U}\CS$};
		\draw (9.13,6.0) node {$\DSB$};
		
		\draw[->,>=stealth] (7.5,4.2) to node [right] {$\TrB$} (7.5,3.3);

		\rotatedSpherePlaneRight{7.5}{1.5}{-70}{0.7}

		\draw[->,>=stealth] (4.3,1.5) to node [above] {$\Psi_{U}^{\CS}$} (5.8,1.5);
		\draw (9.05,1.5) node {$\DS$};
	\end{tikzpicture}
	}}
	\caption{Schematic representation of the role of states and, in particular, the physical domain $\TrB(\DSB\cap\CS)$ in the descriptions of the spaces and maps involved in the present framework.  In the upper two diagrams, the ball represents $\DSB$, the convex set of all system-bath density matrices, and the ruled areas are the intersections $\CS\cap\DSB$ and $\Ad_{U}\CS\cap \DSB$.  In the lower two diagrams, the ball represents the convex set $\DS$ of all system density matrices while the ruled areas represent $\TrB(\CS\cap\DSB)$, which is the physical domain, and $\TrB\Ad_{U}(\CS\cap\DSB)$, which is the image under $\Psi_{U}^{\CS}$ of the physical domain.}
	\label{fig:4}
\end{figure}

We illustrate these concepts in Fig.~\ref{fig:4}.

\subsection{Operator Sum Representations}
\label{sec:OSR}
The map $\Psi_{U}^{\CS}$ may be extended to a $\mathbb{C}$-linear,  $\dag$-linear, trace-preserving map $\Phi_{U}^{\CS}$ on $\BHS$, for example by defining $\Phi_{U}^{\CS}$ to be zero on $\BHS/\TrB\CS$, the orthogonal complement of $\TrB\CS$.  Such an extension, though arbitrary, may be desirable because the resulting $\Phi_{U}^{\CS}$  admits an operator sum representation (OSR) with real-valued coefficients 
\begin{equation}
	\Phi_{U}^{\CS}(X) = \sum_{k}a_{k}E_{k}X E_{k}^{\dag},
\end{equation}
for $a_{k}\in\mathbb{R}$ \cite{dePillis:1967,Choi:1975}.  

In the case of an OSR such that $a_{k} = 1$ for all $k$, we call such a representation a Kraus OSR.

\begin{mylemma}
	The assignment map $\AM_{\CS}:\TrB\CS\to\CS/\CS_{0}$ associated to any $\Sgp$-consistent subspace $\CS$ admits an OSR with real-valued coefficients of the form $\AM_{\CS}(X) = \sum_{k}a_{k}Q_{k}X Q_{k}^{\dag} + \CS_{0}$.
\end{mylemma}
\begin{proof}
	This may be shown, for example, by extending $\AM_{\CS}$ to a $\mathbb{C}$-linear, $\dag$-linear map $\widehat{\AM_{\CS}}:\BHS\to\BHSB/\CS_{0}$ and applying Choi's method to $\widehat{\AM_{\CS}}$ as follows.  First, choose orthonormal bases $\{\ket{i}\}_{i=1}^{d_{\textsc{s}}}$ and $\{\ket{\alpha}\}_{\alpha=1}^{d_{\textsc{b}}}$ for $\HS$ and $\HB$, observe that $\widehat{\AM_{\CS}}\otimes\id_{\BHS}$ is $\dag$-linear so that the affine subspace $\widehat{\AM_{\CS}}\otimes\id_{\BHS}(\ketbra{\mathcal{E}}{\mathcal{E}})$ is spanned by Hermitian operators, where $\ket{\mathcal{E}} = \sum_{i}\ket{i}\otimes\ket{i}\in\HS\otimes\HS$.  Choose any Hermitian operator $A$ in this space, choosing $A\geq 0$ if possible, and eigendecompose 
	\begin{equation}
		\widehat{\AM_{\CS}}\otimes\id_{\BHS}(\ketbra{\mathcal{E}}{\mathcal{E}}) \ni A =  \sum_{k}a_{k}\ketbra{q_{k}}{q_{k}}.
	\end{equation}
	The operators $Q_{k}\in\BL(\HS;\HSB)\simeq \HSB\otimes\HS^{*}$ defined by 
	\begin{equation}
		Q_{k} = \sum_{j,\alpha,i}\braket{j,\alpha,i}{q_{k}}\ketbra{j,\alpha}{i}
	\end{equation}
	yield an OSR for the extended assignment map: $\widehat{\AM_{\CS}}(X) = \sum_{k}a_{k}Q_{k}X Q_{k}^{\dag} + \CS_{0}$.  These $Q_{k}$ are essentially partial transposes of the $\ket{q_{k}}\in\HSB\otimes\HS$, transforming the ``column vectors'' in the second $\HS$ to ``row vectors'' in $\HS^{*}$.
\end{proof}
It follows that $\Phi_{U}^{\CS}(X) = \sum_{k,\alpha}a_{k}E_{k\alpha}X E_{k\alpha}^{\dag}$, where $E_{k\alpha} = \langle \alpha|UQ_{k}\in\BHS$.  

In special cases, this general recipe may not be the most efficient for obtaining an OSR.  For example, in the case of a Kraus map where $\AM_{\CS}(X) = X\otimes\rhoB$ and $\rhoB = \sum_{\alpha}p_{\alpha}\ketbra{\alpha}{\alpha}$, an OSR for $\AM_{\CS}$ is given by the operators $Q_{\alpha} = \sqrt{p_{\alpha}}\identity\otimes\ket{\alpha}$, so that $\Phi_{U}^{\CS}(X) = \sum_{\alpha,\beta}E_{\alpha,\beta}X E_{\alpha,\beta}^{\dag}$ where $E_{\alpha,\beta} = \bra{\beta}UQ_{\alpha}$.  It bears repeating that, regardless of these choices for constructing the OSR, the resulting map $\Phi_{U}^{\CS}$ only acts as the subsystem dynamical map due to $U$ on system states in the ``physical domain'' $\TrB(\DSB\cap\CS)\subset\DS$.

\section{Complete Positivity}
\label{sec:CPity}

Having described, for each allowed joint system-bath unitary evolution $U\in \Sgp$, the unique dynamical map $\Psi_{U}:\TrB\CS\to\BHS$ which is $\mathbb{C}$-linear, $\dag$-linear, and trace-preserving, we turn to the question of complete positivity.

\subsection{Notions of Complete Positivity}
Because $\AM_{\CS}$ and $\Psi_{U}$ are generally defined on a subspace of $\BHS$, rather than the full algebra, we may consider three generally nonequivalent definitions of complete positivity of the subsystem dynamics; one which is essentially the original definition of Stinespring \cite{Stinespring:1955}, and two more that we introduce.
\begin{mydefinition}[Complete Positivity]
	\label{def:CPmap}
	Let $\mathcal{K}$ and $\HH$ be Hilbert spaces and let $\mathcal{R}\subset\BL(\mathcal{K})$ be a self-adjoint $\mathbb{C}$-linear subspace spanned by positive operators.  A $\mathbb{C}$-linear, $\dag$-linear map $F:\mathcal{R}\to\BL(\HH)$ is 
	\begin{enumerate}
	\item \emph{Completely Positive (CP)} \cite{Stinespring:1955} if $F\otimes \id:\mathcal{R}\otimes\BHW\to\BL(\HH\otimes\HW)$ is a positive map for all finite dimensional Hilbert spaces $\HW$, i.e., every positive operator in 
	\begin{equation}
		\mathcal{R}\otimes\BHW = \Span_{\mathbb{C}}\{A\otimes B\;:\; A\in\mathcal{R}, B\in\BHW\}
	\end{equation}
	is mapped to a positive operator in $\BL(\HH\otimes\HW)$.  $\HW$ may be thought of as the state space of a non-interacting, non-evolving ``witness'' system;
	\item \emph{Completely Positively Trace-Preserving Extensible (CPTE)} if $F$ admits a completely positive trace-preserving extension $\tilde{F}:\BL(\mathcal{K})\to\BL(\HH)$, i.e., if there exists a trace-preserving CP map $\tilde{F}:\BL(\mathcal{K})\to\BL(\HH)$ such that $\tilde{F}\big|_{\mathcal{R}} = F$;
	\item \emph{Completely Positively Zero Extensible (CPZE)} if ${F\circ\mathcal{P}_{\mathcal{R}}}:\BL(\mathcal{K})\to\BL(\HH)$ is completely positive, where $\mathcal{P}_{\mathcal{R}}:\BL(\mathcal{K})\to\mathcal{R}$ is the orthogonal projection onto $\mathcal{R}$ with respect to the Hilbert-Schmidt inner product.
	\end{enumerate}
\end{mydefinition}

It is straightforward to see that if a map $F:\mathcal{R}\to\BL(\HH)$ is CPTE, it is also CP and trace-preserving, since restricting the trace-preserving CP extension $\tilde{F}:\BL(\mathcal{K})\to\BL(\HH)$ to the subspace $\mathcal{R}\subset \BL(\mathcal{K})$ does not break either complete positivity or trace-preservation.  The following theorem, translated into this terminology of consistent subspaces, may be considered a partial converse of this statement.

\begin{mytheorem}[Jen{\v{c}}ov{\'a} \cite{Jencova:2012}]
	\label{thm:Arveson}
	Since $\CS$ is spanned by states, so is $\TrB\CS$, and therefore every CP map with domain $\TrB\CS$ can be extended to a CP map on $\BHS$, i.e., every CP map $F:\TrB\CS\to\BHS$ admits a completely positive extension $\tilde{F}:\BHS\to\BHS$ such that $\tilde{F}\big|_{\TrB\CS} = F$.  This is a generalization of Arveson's theorem \cite{Arveson:1969}.
\end{mytheorem}

Jen{\v{c}}ov{\'a}'s theorem is a partial converse of CPTE$\Rightarrow$CP because it does not guarantee that the CP extension of a trace-preserving CP map on $\TrB\CS$ is itself trace-preserving on $\BHS\setminus\TrB\CS$.  Indeed, the conditions allowing trace-preserving CP extensions are not currently known \cite{Heinosaari:2012}.  However, even if the CP extension $\tilde{F}$ is not trace-preserving in general, it may still be important since $\tilde{F}$, being a CP map defined on the entire algebra $\BHS$, admits a Kraus OSR \cite{Kraus:1971, Choi:1975}.  As such, the existence of a Kraus OSR for the map $F:\mathcal{R}\to\BH$ is an alternate characterization of CP-ness, even for maps defined on state-spanned subspaces such as $\TrB\CS$.

Likewise, if a map $F:\mathcal{R}\to\BL(\HH)$ is CPZE, then it is CP, since the zero extension is just one possible extension of $F$.  The following theorem may be considered a partial converse of this statement.

\begin{mytheorem}[Choi \& Effros \cite{Choi:1977}]
	\label{thm:ChoiEffros}
	If $\TrB\CS$ is a unital $C^{\ast}$-subalgebra of $\BHS$, then the orthogonal projection $\mathcal{P}_{\TrB\CS}$ is completely positive, and therefore every CP map with domain $\TrB\CS$ is CPZE.
\end{mytheorem}

That the property of being CP does not necessarily imply CPZE-ness is demonstrated by the following counterexample.

\begin{myexample}
	\label{ex:ChoiEffrosCounterEx}
	Suppose $\HS$ is 3-dimensional with orthonormal basis $\{|i\rangle\}_{i=0}^{2}$.  Fix some $\rho\in\DB$ and let 
	\begin{equation}
		\CS = \Span_{\mathbb{C}}\{\ketbra{i}{j}\otimes\rho\;:\; (0,1)\neq (i,j)\neq (1,0)\}.
	\end{equation}
	The assignment map $\AM_{\CS}$ associated with this $\CS$ is CP but not CPZE, demonstrating that the hypothesis of Theorem \ref{thm:ChoiEffros} that $\TrB\CS$ be a subalgebra of $\BHS$ is needed.
\end{myexample}

\begin{proof}
	$\TrB\CS$ is given by 
	\begin{equation}
		\TrB\CS = \Span_{\mathbb{C}}\{\ketbra{i}{j}\;:\; (0,1)\neq (i,j)\neq (1,0)\}.
	\end{equation}  
	The orthogonal projection onto $\TrB\CS$ has Choi matrix
	\begin{equation}
		\mathrm{Choi}(\mathcal{P}_{\TrB\CS}) = \sum_{\jmdstack{0\leq i,j\leq 2}{(0,1)\neq (i,j)\neq (1,0)}}\ketbra{i}{j}\otimes\ketbra{i}{j}.
	\end{equation}
	This matrix has eigenvector $-\ket{00}-\ket{11}+\sqrt{2}\ket{22}$ with eigenvalue $1-\sqrt{2}<0$, so it is not positive, and therefore $\mathcal{P}_{\TrB\CS}$ is not CP.  This demonstrates that there exist self-adjoint subspaces $\TrB\CS\subset\BHS$ containing the identity and which are not $C^{\ast}$ subalgebras for which the orthogonal projection $\mathcal{P}_{\TrB\CS}$ is not CP.  The zero-extended assignment map $\tilde{\AM} = \AM_{\CS}\circ\mathcal{P}_{\TrB\CS}$ is simply $\tilde{A}(\sigma) = (\mathcal{P}_{\TrB\CS}(\sigma))\otimes\rho$.  Therefore, for any $X\in\BHSW$, 
	\begin{equation}
		\tilde{\AM}\otimes\id_{\BHW}(X) = (\mathcal{P}_{\TrB\CS}\otimes\id(X))\otimes\rho,
	\end{equation}
	so that $\tilde{\AM}$ is not CP because $\mathcal{P}_{\TrB\CS}$ is not CP, and therefore $\AM_{\CS}$ is not CPZE.  However, for any $X\in\TrB\CS\otimes\BHW$, 
	\begin{equation}
		\AM_{\CS}\otimes\id(X) = X\otimes\rho,
	\end{equation}
	which is a positive map for all finite-dimensional witnesses, so that $\AM_{\CS}$ is CP.
\end{proof}

One observation to make about $\Sgp$-consistency, that will be useful in applying these  notions of complete positivity, is that, when a $\Sgp$-consistent subspace $\CS$ is tensored with the operator algebra of a ``witness system'' as in the definition of complete positivity above, the resulting operator space is still appropriately consistent.  In other words, $\Sgp$-consistency is stable in the following sense:
\begin{mylemma}
	\label{lem:GConsistStability}
	If $\Sgp\subset\USB$ and $\CS\subset\BHSB$ is a $\Sgp$-consistent subspace, then for any finite-dimensional Hilbert space $\HW$, $\CS\otimes \BHW\subset \BL(\HS\otimes\HB\otimes \HW)$ is a $\Sgp\otimes\UW$-consistent subspace.  In particular, it is a $\Sgp\otimes\identity_{\HW}$-consistent subspace.
\end{mylemma}
\begin{proof}
	First, observe that, if $\CS\in\BHSB$ is a $\Sgp$-consistent subspace, then it is $\mathbb{C}$-linear and spanned by states.  It follows that $\CS\otimes\BHW$ is $\mathbb{C}$-linear.  Furthermore, since $\CS$ and $\BHW$ are each spanned by states, $\CS\otimes\BHW$ is spanned by states of the form $\rho\otimes\sigma$ with $\rho\in\CS\cap\DSB$ and $\sigma\in\DW$.\\
	\indent It remains only to prove $\Sgp\otimes\UW$-consistency.  To that end, let $X\in\ker\left(\TrB\big|_{\CS\otimes\BHW}\right) = \CS_{0}\otimes \BHW$, $U\in\Sgp$, and $V\in\UW$.  $X$ may be expanded as $X = \sum_{i}Y_{i}\otimes Z_{i}$ where $\{Y_{i}\}\subset\CS_{0}$ and $\{Z_{i}\}\subset\BHW$.  Then
	\begin{align}
		\TrB\big[(U\otimes V)X(U\otimes V)^{\dag}\big] & = \sum_{i}\TrB(U Y_{i}U^{\dag})\otimes (V Z_{i}V^{\dag}) \nonumber\\
		& = 0
	\end{align}
	 since $Y_{i}\in\CS_{0}$ and $\CS$ is $\Sgp$-consistent.  Then $(U\otimes V)X(U\otimes V)^{\dag}\in\ker\left(\TrB\big|_{\BHSBW}\right)$, so $\CS\otimes\BHW$ is $\Sgp\otimes\UW$-consistent.  It follows trivially that $\CS\otimes\BHW$ is $\Sgp\otimes\identity_{\HW}$-consistent.
\end{proof}

\subsection{Defining Completely Positive Assignment Maps}
\label{sec:CPassignment}
Let $\CS$ be a $\Sgp$-consistent subspace, let $\mathcal{V}_{0} = \ker\big(\TrB\big|_{\CS}\big)$, and let $\HW$ be a finite-dimensional Hilbert space.  We would like to apply the above notions of complete positivity not only to the dynamical maps $\Psi_{U}$, but also to the assignment map $\AM_{\CS}:\TrB\CS\to\CS/\CS_{0}$.  To do so, we will need to identify the positive elements within $\CS/\CS_{0}\otimes \BL(\HH_{W})$.  To that end, we begin by establishing a natural isomorphism between $\CS/\CS_{0}\otimes \BL(\HH_{W})$ and a space in which the positive elements are readily identifiable.
\begin{mylemma}
	\label{lem:tensoredIsomorphism}
	The space $\CS/\CS_{0}\otimes\BHW$ is naturally isomorphic to $\CS\otimes\BHW/\CS_{0}\otimes\BHW$.
\end{mylemma}
\begin{proof}
We may construct a natural isomorphism 
	\begin{equation}
		h:\big(\CS\otimes\BHW\big)/\big(\CS_{0}\otimes\BHW\big) \to \CS/\CS_{0}\otimes\BHW
	\end{equation}
	as follows.  Let $p:\CS\to\CS/\CS_{0}$ be the natural projection map and let $\tilde{h} := p\otimes\id:\CS\otimes\BHW\to\CS/\CS_{0}\otimes \BHW$ be the $\mathbb{C}$-linear map such that $\tilde{h}(Y\otimes Z) = p(Y)\otimes Z$ for any $Y\in\CS$ and $Z\in\BHW$.  Any $X\in\CS\otimes\BHW$ can be written $X = \sum_{i}Y_{i}\otimes Z_{i}$, where the $\{Z_{i}\}\subset \BHW$ are linearly independent.  Then $\tilde{h}(X) = \sum_{i}p(Y_{i})\otimes Z_{i}$ so $\tilde{h}(X) = 0$ if and only if $p(Y_{i}) = 0$ for all $i$, i.e., if and only if $Y_{i}\in\CS_{0}$ for all $i$.  Thus $\ker\tilde{h} = \CS_{0}\otimes\BHW$, and since $\tilde{h}$ is clearly surjective, we can define the isomorphism $h:\CS\otimes\BHW/\CS_{0}\otimes\BHW\to\CS/\CS_{0}\otimes\BHW$ by the commutative diagram in Figure \ref{fig:isomorphismCD}.

	\begin{figure}
		\tikzimageswap{\includegraphics{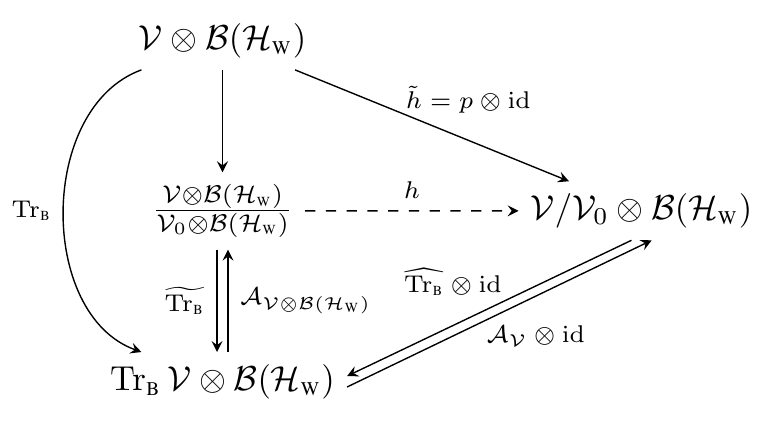}}{
		\tikzsetnextfilename{figure5}
		\begin{tikzpicture}
			\matrix (m) [matrix of math nodes, ampersand replacement=\&,row sep=3em,column sep=5em,minimum width=2em]
			{
				\CS\otimes\BHW \&  \\
		  		\frac{\CS\otimes\BHW}{\CS_{0}\otimes\BHW} \& \CS/\CS_{0}\otimes\BHW\\
		  		\TrB\CS\otimes\BHW \& \\
		  	};
		  	\path[-stealth,font=\scriptsize]
			(m-1-1) edge node [left] {} (m-2-1)
			(m-1-1) edge node [above] {$\hspace{3em}\tilde{h}=p\otimes\id$} (m-2-2)
			(m-2-1) edge [dashed] node [above] {$h$} (m-2-2)
			([xshift=-0.35ex]m-2-1.south) edge node [left] {$\widetilde{\TrB}$} ([xshift=-0.35ex]m-3-1.north)
			([xshift=0.35ex]m-3-1.north) edge node [right] {$\AM_{\CS\otimes\BHW}$} ([xshift=0.35ex]m-2-1.south)
			(m-1-1) edge [bend right=70] node [left] {$\TrB$} (m-3-1)
			([yshift=-0.35ex]m-3-1.east) edge [bend right=0] node [below] {$\hspace{3em}\AM_{\CS}\otimes\id$} ([xshift=0.65ex]m-2-2.south)
			([xshift=-0.65ex]m-2-2.south) edge [bend left=0] node [above] {$\hspace{-3em}\widehat{\TrB}\otimes\id$} ([yshift=0.35ex]m-3-1.east);
		\end{tikzpicture}}
		\caption{Commutative diagram defining the isomorphism $h:\CS\otimes\BHW/\CS_{0}\otimes\BHW\to\CS/\CS_{0}\otimes\BHW$.  Also depicted is the relationship between the assignment maps $\AM_{\CS\otimes\BHW}$ and $\AM_{\CS}$, namely, $\AM_{\CS}\otimes\id_{\BHW} = h\circ \AM_{\CS\otimes\BHW}$. $\widetilde{\TrB}$ and $\widehat{\TrB}$ are the unique maps satisfying the determination problems of $\Tr_{B}$ by the projections $\CS\otimes\BHW\to\CS\otimes\BHW/\CS_{0}\otimes\BHW$ and $\CS\to\CS/\CS_{0}$, respectively.}
		\label{fig:isomorphismCD}
	\end{figure}	
\end{proof}
Furthermore,  it is easy to see that if $\AM_{\CS}$ and $\AM_{\CS\otimes\BHW}$ are the assignment maps associated with $\CS$ and $\CS\otimes\BHW$, then $ \AM_{\CS}\otimes\id_{\BHW} = h\circ\AM_{\CS\otimes\BHW}$, where $h:\big(\CS\otimes\BHW\big)/\big(\CS_{0}\otimes\BHW\big) \to \CS/\CS_{0}\otimes\BHW$ is the isomorphism described in Lemma \ref{lem:tensoredIsomorphism}.

Now, observe that 
\begin{equation}
	\ker\left(\TrB\big|_{\CS\otimes\BHW}\right) = \CS_{0}\otimes\BHW.
\end{equation}  
By invoking Lemma \ref{lem:tensoredIsomorphism}, in the same way that we defined a positive cone $(\CS/\CS_{0})^{+}$ comprising those affine subspaces in $\CS/\CS_{0}$ that each contain at least one positive operator, we are now able to describe a positive cone in $\CS/\CS_{0}\otimes \BHW$, thereby defining a  matrix ordering \cite{Choi:1977} for $\CS/\CS_{0}$.
\begin{mydefinition}
	An element in $\CS/\CS_{0}\otimes\BHW$ will be considered positive, and therefore belonging to $(\CS/\CS_{0}\otimes\BHW)^{+}$, if the corresponding affine subspace in $\big(\CS\otimes\BHW\big)/\big(\CS_{0}\otimes\BHW\big)$ contains at least one positive operator.
\end{mydefinition}
\begin{mylemma}
	This matrix ordering of $\CS/\CS_{0}$ (and the analogous matrix ordering for $\BHSB/\ker \TrB$) is such that all maps in Figure \ref{fig:commutativeDiagram2} are completely positive, with the exception of $\AM_{\CS}$ and $\Psi_{U}^{\CS}$.  Indeed, it is the minimal matrix ordering necessary for this (minimal in the sense that the fewest elements of $\CS/\CS_{0}\otimes\BHW$ are considered ``positive'').
\end{mylemma}
\begin{proof}
	Fix some finite-dimensional $\HW$ and let $B\in(\CS/\CS_{0}\otimes\BHW)^{+}$.  It follows that $h^{-1}(B)$ is positive in $\CS\otimes\BHW/\CS_{0}\otimes\BHW$ and therefore that there exists a positive $A\in\CS\otimes\BHW$ such that $p\otimes\id(A) = B$, where $p:\CS\to\CS/\CS_{0}$ is the natural projection map.  So the positive cone $(\CS/\CS_{0}\otimes\BHW)^{+}$ is precisely the image through $p\otimes\id$ of the positive cone $(\CS\otimes\BHW)^{+}$.  Consequently, $(\CS/\CS_{0}\otimes\BHW)+$ is the minimal positive cone (for all $\HW$) to make $p:\CS\to\CS/\CS_{0}$ a completely positive map.  As we have given $\BHSB/\ker\TrB$ the analogous matrix ordering, the same is true of the natural projection $p':\BHSB\to\BHSB/\ker\TrB$.  Thus, minimality is proved, and it only remains to examine complete positivity of the other maps. 
	
	Suppose $B\in\CS/\CS_{0}\otimes\BHW$ is positive.  Then there exists $A\in\CS\otimes\BHW$ which is positive and is such that $p\otimes\id(A) = B$.  Complete positivity of $\Ad_{U}$ and $p'$ imply that $(p'\otimes\id)\circ(\Ad_{U}\otimes\id)(A) = (p'\circ\Ad_{U})\otimes\id(A)$ is positive.  But $(p'\circ\Ad_{U})\otimes\id(A) = (\Ad_{U}^{0}\circ p)\otimes\id(A) = \Ad_{U}^{0}\otimes\id(B)$, so that $\Ad_{U}^{0}\otimes\id$ is a positive map for all $\HW$, and therefore $\Ad_{U}^{0}$ is completely positive.  Likewise, $\TrB\otimes\id(A) = \widehat{\TrB}\otimes\id(B)$ is positive, so that $\widehat{\TrB}$ is completely positive.  And complete positivity of $\TrB^{0}$ is proved analogously: for any positive $B\in\BHSB/\ker\TrB$ there exists positive $A\in (p'\otimes\identity)^{-1}(B)$, and therefore $\TrB^{0}\otimes\id(B) = \TrB\otimes\id(A)$ is positive, so that $\TrB{0}$ is completely positive.
\end{proof}

Then $\AM_{\CS}:\TrB\CS\to\CS/\CS_{0}$ is completely positive if and only if $\AM_{\CS}\otimes\id_{\BHW}:(\TrB\CS)\otimes\BHW\to\CS/\CS_{0}\otimes\BHW$ is positive for all finite-dimensional $\HW$, i.e., if and only if $\AM_{\CS}\otimes\id_{\BHW}$ maps positive operators in $\TrB(\CS)\otimes\BHW$ to elements in $(\CS/\CS_{0}\otimes\BHW)^{+}$.  Equivalently, $\AM_{\CS}$ is completely positive if and only if $\TrB\big(\DSBW\cap\big[\CS\otimes\BHW\big]\big) = \DSW\cap\big[(\TrB\CS)\otimes\BHW\big]$ for all finite dimensional $\HW$, i.e., if and only if every system-witness state in $(\TrB\CS)\otimes \BHW$ is ``covered'' by a system-bath-witness state in $\CS\otimes\BHW$.

We summarize the different notions of complete positivity and their interrelations in Fig.~\ref{fig:Venn-CP}.

\begin{figure}[t]
	\scalebox{0.8}{
	\tikzimageswap{\includegraphics{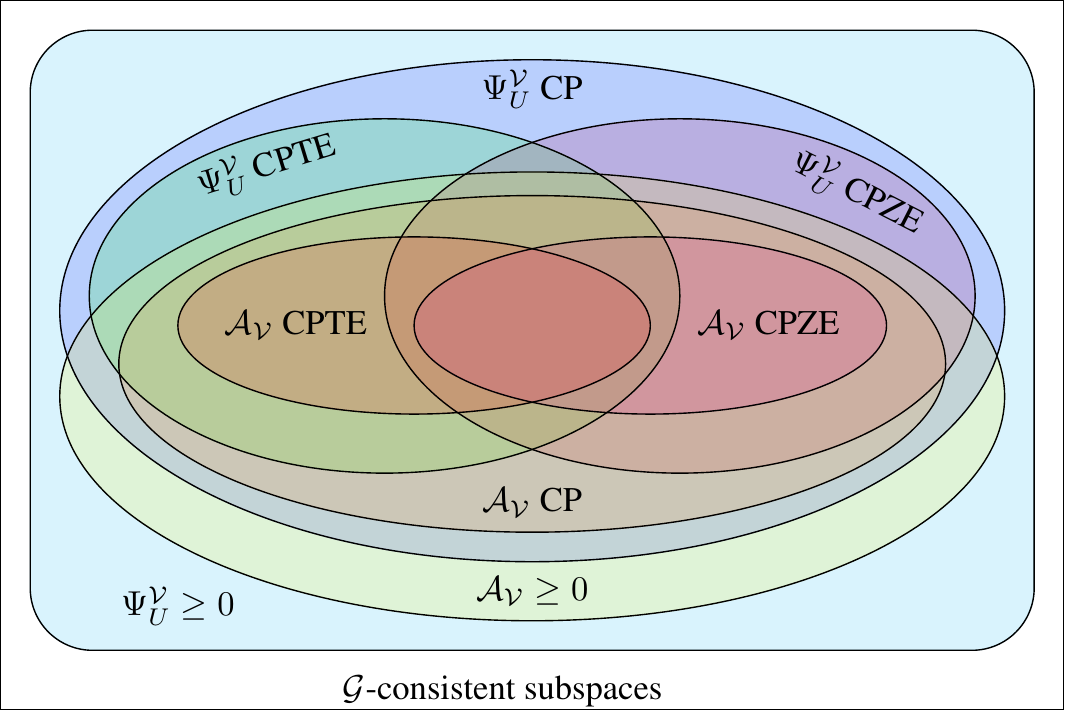}}{
	\tikzsetnextfilename{figure6}
	\begin{tikzpicture}[scale=.6]
		\draw (0,0) rectangle (18, 12);
		\draw (8.5,0.3) node {$\Sgp$-consistent subspaces};
		\filldraw[rounded corners=4ex, opacity = 0.15,  color = cyan] (.5,1) rectangle (17.5,11.5);
		\filldraw [opacity = 0.15,  color = blue] (9,6.75) ellipse [x radius=8, y radius=4.25]; 
		\filldraw [opacity = 0.15,  color = green] (6.5,7) ellipse [x radius=5., y radius=3]; 
		\filldraw [opacity = 0.15,  color = purple] (11.5,7) ellipse [x radius=5., y radius=3]; 
		\filldraw [opacity = 0.15,  color = yellow] (9.,5.3) ellipse [x radius=8, y radius=3.8]; 
		\filldraw [opacity = 0.15,  color = orange] (9,5.85) ellipse [x radius=7, y radius=2.85]; 
		\filldraw [opacity = 0.15,  color = red] (7,6.5) ellipse [x radius=4., y radius=1.5]; 
		\filldraw [opacity = 0.15,  color = magenta] (11,6.5) ellipse [x radius=4., y radius=1.5]; 
		\draw[rounded corners=4ex] (.5,1) rectangle (17.5,11.5);
		\draw (3.0,1.75) node {$\Psi_{U}^{\CS}\geq 0$};
		\draw (9,6.75) ellipse [x radius=8, y radius=4.25]; 
		\draw (9,10.5) node[rotate=0] {$\Psi_{U}^{\CS}$ CP};
		\draw (6.5,7) ellipse [x radius=5., y radius=3]; 
		\draw (4.5,9.2) node[rotate=17] {$\Psi_{U}^{\CS}$ CPTE};
		\draw (11.5,7) ellipse [x radius=5., y radius=3]; 
		\draw (14.5,8.75) node[rotate=-27] {$\Psi_{U}^{\CS}$ CPZE};
		\draw (9.,5.3) ellipse [x radius=8, y radius=3.8]; 
		\draw (9,2.) node[rotate=0] {$\AM_{\CS}\geq 0$};
		\draw (9,5.85) ellipse [x radius=7, y radius=2.85]; 
		\draw (9,3.5) node[rotate=0] {$\AM_{\CS}$ CP};
		\draw (7.,6.5) ellipse [x radius=4., y radius=1.5]; 
		\draw (5,6.5) node[rotate=0] {$\AM_{\CS}$ CPTE};
		\draw (11,6.5) ellipse [x radius=4., y radius=1.5]; 
		\draw (13,6.5) node {$\AM_{\CS}$ CPZE};
	\end{tikzpicture} } }
	\caption{Venn diagram of the set of all $\Sgp$-consistent subspaces $\CS$ (for some fixed $\Sgp$) satisfying (complete) positivity conditions for the dynamical maps $\Psi_{U}^{\CS}$ or for the assignment map $\AM_{\CS}$.  Note that the sets of $\Sgp$-consistent subspaces satisfying (complete) positivity for the dynamical maps $\Psi_{U}^{\CS}$ are the sets of spaces which satisfy the stated condition for all $U\in\Sgp$.  It is not known whether all of the indicated subsets are non-empty.   For example, it is not known if there exist $\Sgp$-consistent subspaces $\CS$ for which the assignment map $\AM_{\CS}$ is positive but not CP.}
	\label{fig:Venn-CP}
\end{figure}

Determining which $\Sgp$-consistent subspaces give rise to completely positive (CP, CPTE, or CPZE) assignment maps is a rich and complex problem.  However, as we will discuss later, when $\TrB\CS = \BHS$, a theorem of Pechukas \cite{Pechukas:1994} generalized by Jordan, \textit{et al.}, \cite{Jordan:2004}, shows that the assignment map for a $\USB$-consistent subspace $\CS$ is positive if and only if it is CPZE and $\CS$ is of the form $\CS = \BHS\otimes\rhoB$ for some fixed $\rhoB\in\DB$.  We now show that in the simple case where $\TrB\CS$ is one-dimensional, the question of complete positivity of the assignment map may be answered comprehensively.

\begin{mylemma}
	\label{lem:oneDimTrBV}
	If $\CS\subset\BHSB$ is a $\Sgp$-consistent subspace such that $\TrB\CS$ is 1-dimensional, then the assignment map $\AM_{\CS}:\TrB\CS\to\CS/\CS_{0}$ is CPZE.
\end{mylemma}
\begin{proof}
	Let $\rho\in\DSB\cap\CS$ and $\rhoS=\TrB \rho = \DS\cap\TrB\CS$, so that $\TrB\CS = \mathbb{C}\rhoS$.  The assignment map is then $\AM_{\CS}(z\rhoS) = z\rho + \CS_{0}$ for any $z\in\mathbb{C}$.  Let $\tilde{\AM}:\BHS\to\BHSB/\CS_{0}$ be the zero-extension of $\AM_{\CS}$, i.e., for any $X\in\BHS$
	\begin{equation}
		\tilde{\AM}(X) = \frac{\Tr(\rhoS X)}{\Tr(\rhoS^{2})}\rho + \CS_{0}.
	\end{equation}
	Let $\HW$ be a finite-dimensional witness Hilbert space.  Then, for any  $Y\in\BHSW$, 
	\begin{align}
		h^{-1}\circ\tilde{\AM}\otimes\id(Y) & = \rho\otimes \frac{\Tr_{\textsc{s}}[(\sqrt{\rhoS}\otimes\identity)Y(\sqrt{\rhoS}\otimes\identity)]}{\Tr(\rhoS^{2})}\nonumber\\
		& \qquad  + \CS_{0}\otimes\BHW,
	\end{align}
	where $h:\BHSB\otimes\BHW/\CS_{0}\otimes\BHW\to\BHSB/\CS_{0}\otimes\BHW$ is the isomorphism constructed analogously to that in Lemma \ref{lem:tensoredIsomorphism}.  It follows that $\tilde{\AM}\otimes \id_{\BHW}$ is a positive map for any finite-dimensional witness $\HW$.  Thus, $\AM_{\CS}$ is CPZE.
\end{proof}
It is an open problem to find some generalizations of these results to arbitrary $\Sgp$-consistent subspaces.  In particular, Example \ref{ex:ChoiEffrosCounterEx} shows that Lemma \ref{lem:oneDimTrBV} does not extend to higher dimensional $\TrB\CS$ without modification.

\subsection{Why Complete Positivity?}
\label{sec:whyCP}
Definition~\ref{def:CPmap} advocates different notions of complete positivity. In order to properly address the question of complete positivity of system dynamics, we must first ask ourselves: what aspects of complete positivity make it a physically interesting property?  There may be several valid answers to this question and it is possible that there will be no general consensus as to which is most compelling.  Here are two:
\begin{enumerate}
	\item The existence of a non-interacting, non-evolving ``witness'' (for which the initial system-witness state can be entangled) should not be enough to break the positivity of the evolution.  This is frequently claimed as a mandate for complete positivity.\label{item:witnessInterp}
	\item As mentioned above, CP maps are sometimes representable by Kraus OSRs $\Phi(\rho) = \sum E_{i}\rho E_{i}^{\dag}$, and Kraus OSRs are ubiquitous in quantum information theory.  For example, such an OSR can be thought of as a convex combination $\Phi(\rho)=\sum p_{i}\rho_{i}$, where $p_{i} = \Tr(E_{i}\rho E_{i}^{\dag})$ and $\rho_{i} = E_{i}\rho E_{i}^{\dag}/p_{i}$.  As such, it offers a convenient interpretation as a non-selective ``measurement'' of the system by the environment, yielding outcome state $\rho_{i}$ with probability $p_{i}$.\label{item:krausOSRInterp}
\end{enumerate}

\subsubsection{Witnessed Positivity}
If the importance of complete positivity in quantum dynamics is to do with maintaining positivity in the presence of a witness, i.e., point 1 above, then the key question, as we now show, is whether the assignment map is CP.  If we consider some dynamical map $\Psi_{U}^{\CS}:\TrB\CS\to\BHS$ for $U\in \Sgp$ and a finite-dimensional witness Hilbert space $\HW$, then 
\begin{equation}
	\Psi_{U}^{\CS}\otimes\id_{\BHW} = \TrB\circ\Ad_{U\otimes\identity_{\HW}}\circ(\AM_{\CS}\otimes \id_{\BHW})
\end{equation}
is the dynamical map $\Psi_{U\otimes\identity}^{\CS\otimes\BHW}$ for the $\Sgp\otimes\identity_{\HW}$-consistent subspace $\CS\otimes\BHW$ (see Lemma \ref{lem:GConsistStability}).  
If $\Psi_{U}\otimes\id_{\BHW}$ is positive for all finite-dimensional $\HW$, then, mathematically (by definition \ref{def:CPmap}), it is considered completely positive.  However, it is possible that there exists a witness state space $\HW$ such that  $\Psi_{U}\otimes\id_{\BHW}$ is positive, while $\AM_{\CS}\otimes\id_{\BHW}$ is non-positive, as illustrated in Fig.~\ref{fig:hiddenNonpositivity}. 

\begin{figure}[t]
	\tikzimageswap{\includegraphics{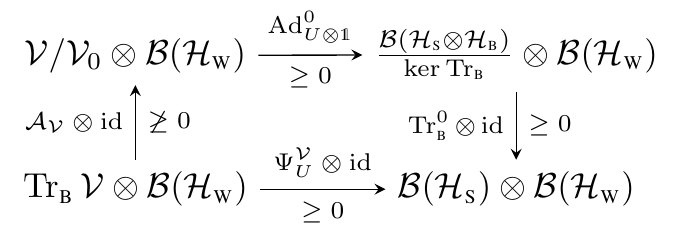}}{
	\tikzsetnextfilename{figure7}
	\begin{tikzpicture}
		\matrix (m) [matrix of math nodes, ampersand replacement=\&,row sep=2em,column sep=3em,minimum width=2em]
		{
	  		\CS/\CS_{0}\otimes\BHW \& \frac{\BHSB}{\ker\TrB}\otimes\BHW \\
	  		\TrB\CS\otimes\BHW \& \BHS\otimes\BHW\\ 
	  	};
	  	\path[-stealth,font=\scriptsize]
		(m-1-1) edge node [above] {$\Ad_{U\otimes\identity}^{0}$} node [below] {$\geq 0$} (m-1-2)
		(m-2-1) edge node [above] {$\Psi_{U}^{\CS}\otimes\id$} node [below] {$\geq 0$} (m-2-2)
		(m-2-1) edge node [left] {$\AM_{\CS}\otimes\id$} node [right] {$\not\geq 0$} (m-1-1)
		(m-1-2) edge node [left] {$\TrB^{0}\otimes\id$} node [right] {$\geq 0$} (m-2-2);
	\end{tikzpicture}}
	\caption{It is possible for non-positivity in the assignment map (possibly tensored with the identity on a witness matrix algebra) to be hidden by the positivity of $\Ad_{U}$ and $\TrB$ for all $U\in\Sgp$, resulting in dynamical maps which appear positive or completely positive according to the mathematical definitions.  However, strictly speaking, these dynamical maps do not satisfy the first answer to the question "why complete positivity?".  An example is considered in Example~\ref{ex:consistentPosCounterEx}.}
	\label{fig:hiddenNonpositivity}
\end{figure}

\begin{observation}
	In such a situation (where $\Psi_{U}\otimes\id_{\BHW} \geq 0$ but $\AM_{\CS}\otimes\id_{\BHW}\ngeq 0$), the positivity of $\Psi_{U}\otimes\id_{\BHW}$ should be considered ``non-physical'' as it arises from assigning some states in $\TrB\CS\otimes\BHW$ to non-positive operators in $\CS\otimes\BHW$ and ``evolving'' those non-positive operators.  If the desire was to show that all system-witness states in $\TrB\CS\otimes\BHW\cap\DSW$ evolve to states in $\DSW$, then we have not fulfilled this goal.  It is only satisfied if the assignment map $\AM_{\CS}\otimes\id_{\BHW}$ is itself a positive map.  This obviously holds for all finite dimensional witnesses $\HW$ if and only if $\AM_{\CS}$ is completely positive (CP).
\end{observation}

We now give an explicit example of such non-physical complete positivity. This is a special case of a class of examples considered in \cite{Rodriguez-Rosario:2010} to illustrate non-positive assignment maps.
\begin{myexample}
	\label{ex:consistentPosCounterEx}
	Let $d_{\textsc{s}} = \dim\HS$, $d_{\textsc{b}} = \dim\HB$, $\rho = \frac{1}{d_{\textsc{s}}d_{\textsc{b}}}\identity$ and $\sigma=\ketbra{\psi}{\psi}\otimes\ketbra{\phi}{\phi}$ for any states $\ket{\psi}\in\HS$ and $\ket{\phi}\in\HB$.  Then $\CS = \Span_{\mathbb{C}}\{\rho, \sigma\}$ is a $\USB$-consistent subspace.  The corresponding assignment map $\AM_{\CS}$ is not positive (see also Fig.~\ref{fig:examplePhysicalDomain}), but $\Psi_{U} = \TrB\circ\Ad_{U}\circ\AM_{\CS}$ is CPZE for all $U\in\USB$.
\end{myexample}
\begin{proof}
	First, note that $\AM_{\CS}$ is not positive, because for $a\geq 0$ and $-\frac{a}{d_{\textsc{s}}}\leq b< -\frac{a}{d_{\textsc{s}}d_{\textsc{b}}}$, $a\rho+b\sigma\ngeq 0$ but $\TrB[a\rho+b\sigma] = \frac{a}{d_{\textsc{s}}}\identity + b\ketbra{\psi}{\psi}\geq 0$.  

	On the other hand, let $\mathcal{P}:\BHS\to\TrB\CS$ be the orthogonal projection onto $\TrB\CS$ and let $\{\ket{1},\dots, \ket{d_{\textsc{s}}}\}$ be an orthonormal basis for $\HS$ with $\ket{1} = \ket{\psi}$.  Then $\mathcal{P}(\ketbra{i}{j}) = 0$ for $i\neq j$, $\mathcal{P}(\ketbra{1}{1}) = \ketbra{1}{1}$, and $\mathcal{P}(\ketbra{i}{i}) = (\identity - \ketbra{1}{1})/(d_{\textsc{s}}-1)$ for $i\neq 1$.  For any $U\in\USB$, the Choi matrix for the map $\Psi_{U}^{\CS}\circ\mathcal{P}$ is
	\begin{equation}
		\TrB(U\sigma U^{\dag})\otimes\ketbra{1}{1} + \frac{\identity - \TrB(U\sigma U^{\dag})}{d_{\textsc{s}}-1}\otimes (\identity - \ketbra{1}{1})
	\end{equation}
	because $\ketbra{1}{1}$ is mapped to $\sigma$ via the assignment map and  $\identity - \ketbra{1}{1}$ is mapped to $d_{\textsc{s}}\rho-\sigma$.  Since $\TrB(U\sigma U^{\dag})\in\DB$, $\identity - \TrB(U\sigma U^{\dag})$ is also positive, so this Choi matrix is obviously positive.  Therefore $\Psi_{U}\circ\mathcal{P}$ is CP, and $\Psi_{U}$ is CPZE.
	\begin{figure}
		\scalebox{0.95}{
		\tikzimageswap{\includegraphics{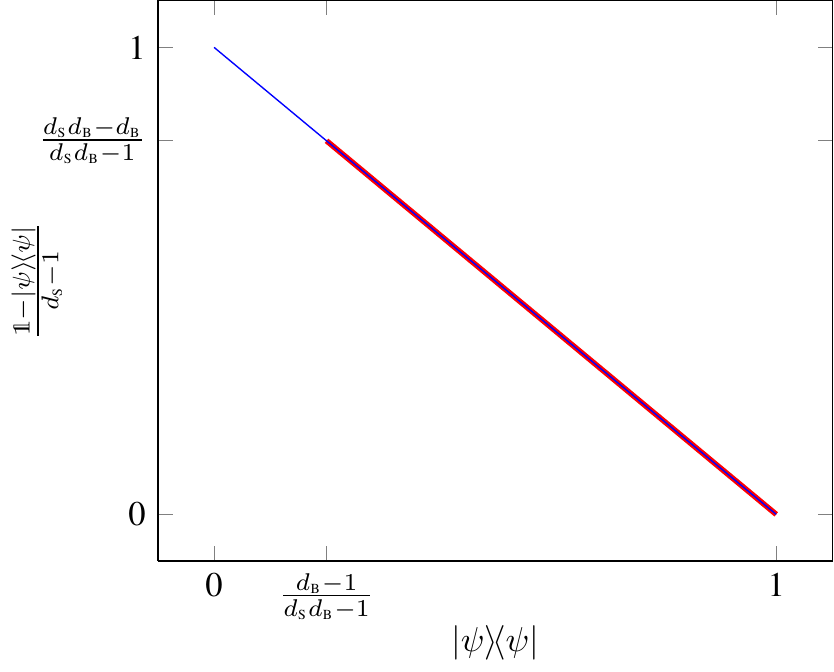}}{
		\tikzsetnextfilename{figure8}
		\begin{tikzpicture}
		    \begin{axis}[
		        xlabel=$\ketbra{\psi}{\psi}$,
		        ylabel=$\frac{\identity-\ketbra{\psi}{\psi}}{d_{\textsc{s}}-1}$,
		        xtick={0,0.2,1},
		        ytick={0,0.8,1},
		        xticklabels = {0, $\frac{d_{\textsc{b}}-1}{d_{\textsc{s}}d_{\textsc{b}}-1}$, 1},
		        yticklabels = {0,$\frac{d_{\textsc{s}}d_{\textsc{b}}-d_{\textsc{b}}}{d_{\textsc{s}}d_{\textsc{b}}-1}$, 1}
		    ]
		    \addplot[ultra thick, color=red] coordinates {
		        (0.2,0.8)
		        (1,0)
			};
		    \addplot[color=blue] coordinates {
		        (0,1)
		        (1,0)
			};
		    \end{axis}
		\end{tikzpicture}}}
		\caption{The system states $\DS\cap\TrB\CS$ for the consistent subspace $\CS$ described in Lemma \ref{ex:consistentPosCounterEx} are convex combinations of the states $\ketbra{\psi}{\psi}$ and $(\identity-\ketbra{\psi}{\psi})/(d_{\textsc{s}}-1)$, indicated by the thin blue line in this figure.  The physical domain $\TrB(\DSB\cap\CS)$, indicated by the thick red line, comprises convex combinations of $\ketbra{\psi}{\psi}$ and $(d_{\textsc{b}}\identity - \ketbra{\psi}{\psi})/(d_{\textsc{s}}d_{\textsc{b}}-1)$.  In particular, $\AM_{\CS}\big[(\identity-\ketbra{\psi}{\psi})/(d_{\textsc{s}}-1)\big] = (d_{\textsc{s}}\rho-\sigma)/(d_{\textsc{s}}-1)$ is a unit trace, \emph{non-positive} operator in $\CS$ and therefore not a state.}
		\label{fig:examplePhysicalDomain}
	\end{figure}
\end{proof}

\subsubsection{Kraus OSR for Dynamical Maps}
On the other hand, if the existence of a Kraus OSR for the dynamical maps (point \ref{item:krausOSRInterp} above) is of primary importance, then the key question is not whether $\Psi_{U}^{\CS}$ is CP, but rather whether $\Psi_{U}^{\CS}$ is completely positively extensible (CPE) for each $U\in\Sgp$. 
This is because 
\begin{mylemma}
	For any $\Sgp$-consistent subspace $\CS$ and any $U\in\Sgp$, the dynamical map $\Psi_{U}^{\CS}$ admits a Kraus OSR if and only if $\Psi_{U}^{\CS}$ is CPE.
\end{mylemma}
\begin{proof}
	If $\Psi_{U}^{\CS}$ is CPE, the CP extension $\hat{\Psi}_{U}$ is defined over the entire algebra $\BHS$, so that Choi's theorem \cite{Choi:1975} may be invoked to conclude that $\hat{\Psi}_{U}$ admits a Kraus OSR.  Likewise, if $\Psi_{U}$ may be written in terms of a Kraus OSR, then that OSR defines a CP extension $\hat{\Psi}_{U}:\BHS\to\BHS$ so that $\Psi_{U}$ is CPE.  
\end{proof}
It should be noted that, while the map $\Psi_{U}$ is uniquely defined by the choice of $\Sgp$-consistent subspace $\CS$ and the choice of $U\in\Sgp$, neither its CP extension $\hat{\Psi}_{U}$ (if one exists) nor the Kraus OSR thereof is uniquely defined.  It should also be noted that any Kraus OSR obtained in this way is still subject to the restrictions mentioned earlier: it cannot be meaningfully applied to any state outside of the physical domain $\TrB(\DSB\cap\CS)$.

\section{Examples}
\label{sec:Examples}

In this section we apply the general framework developed above to earlier work on the topic of complete positivity. In many cases we reinterpret this earlier work in light of the notion of $\Sgp$-consistency, CP, CPTE, and CPZE maps.

\subsection{Kraus (1971)}

The standard Kraus map \cite{Kraus:1971} is built from a $\USB$-consistent subspace of the form $\CS = \BHS\otimes\rhoB$ for some fixed bath state $\rhoB\in\DB$, so that the admissible initial states $\rhoSB\in\{\rhoS\otimes\rhoB\;:\;\rhoS\in\DS\} = \DSB\cap\CS$ are uncorrelated.  As was mentioned at the end of Section~\ref{sec:OSR}, the assignment map $\AM:\TrB\CS\to\CS/\CS_{0}$ associated with this subspace $\CS$ admits an OSR $\AM(X) = \sum_{\nu}Q_{\nu}X Q_{\nu}^{\dag}$ with operators $Q_{\nu} = \sqrt{\lambda_{\nu}}\identity\otimes\ket{\nu}$ where $\rhoB = \sum_{\nu}\lambda_{\mu}\ketbra{\nu}{\nu}$ is the eigendecomposition of $\rhoB$. This OSR is evidently completely positive. It follows that, for any joint unitary evolution $U\in\USB$, the subdynamical map $\Psi_{U}^{\CS}$ is completely positive and admits the OSR
\begin{subequations}
\label{Kraus}
\begin{align}
	\rhoS(t) &= \Psi_{U}^{\CS}[\rhoS(0)]=\sum_{\mu,\nu}E_{\mu\nu} \rhoS(0) E_{\mu\nu}^\dag\\
	\mbox{for   } E_{\mu\nu} &= \bra{\mu}UQ_{\nu} = \sqrt{\lambda_\nu}\bra{\mu}U\ket{\nu}, 
\end{align}
\end{subequations}
where the operation elements $E_{\mu\nu}$ are the Kraus operators.  It is worth noting that, in this case, the ``physical domain'' includes all system states, i.e., $\TrB(\DSB\cap\CS) = \DS$.  A theorem of Pechukas \cite{Pechukas:1994}, the proof of which was supplied by Pechukas in the case of a 2-dimensional system, and by Jordan, Shaji, and Sudarshan \cite{Jordan:2004} for general finite-dimensional systems, shows that the only $\USB$-consistent subspaces exhibiting this property (that $\TrB(\DSB\cap\CS) = \DS$) are those of the form $\CS = \BHS\otimes\rhoB$ for a \emph{fixed} $\rhoB\in\DB$.

\subsection{Pechukas (1994)}
To illustrate the fact that subsystem dynamics need not be completely positive, Pechukas offered an example involving a two-state system and bath and a $\USB$-consistent subspace defined by a $\mathbb{C}$-linear, $\dag$-linear assignment map of the form
\begin{equation}
	\rhoS\mapsto \big[\rhoS\big(\rhoS^{\mathrm{eq}}\big)^{-1}\rhoSB^{\mathrm{eq}} + \rhoSB^{\mathrm{eq}}\big(\rhoS^{\mathrm{eq}}\big)^{-1}\rhoS\big]/2
\end{equation}
where $\rhoS^{\mathrm{eq}} = \TrB(\rhoSB^{\mathrm{eq}})$ are fixed ``equilibrium states'' and $\rhoS^{\mathrm{eq}}>0$ (not just $\geq 0$).
If $\rhoSB^{\mathrm{eq}}$ is \emph{not} a tensor product state, and the domain of this map is taken to be all of $\BHS$, then the image will be a well-defined $\USB$-consistent subspace of $\BHSB$.  However, the assignment map cannot be positive under these assumptions, let alone completely positive, because pure states $\rhoS$ are not mapped to product states $\rhoS\otimes\rhoB$ as a positive assignment map must do (this last point may be seen as a consequence of the fact that a bipartite state $\rhoSB$ with a pure reduced state $\rhoS$ possesses zero mutual information, i.e., saturates subadditivity of the von Neumann entropy \cite{Araki:1970}, and must therefore be a product state \cite{Bruch:1970}).

\subsection{Alicki (1995)}
\label{sec:Alicki}
Commenting on Pechukas' theorem concerning complete positivity, Alicki \cite{Alicki:1995} suggested abandoning either consistency (i.e. $\TrB\circ\AM_{\CS} =\id$) or positivity of the assignment map, rather than giving up complete positivity of the dynamical maps.  To illustrate the possible loss of consistency, he describes a scheme for turning a CP map $T:\BHS\to\BHS$ into a CP assignment map with domain given by the fixed points of $T$.  Suppose $T$ is represented as $T(\rhoS) = \sum_{n}V_{n}\rhoS V_{n}^{\dag}$ and fix some $\rhoSB^{\mathrm{eq}}\in\DSB$.  Then one can postulate a $\mathbb{C}$-linear assignment map
\begin{subequations}
\begin{align}
	\AM_{\CS}(\rhoS) & = \sum_{n}V_{n}\rhoS V_{n}^{\dag}\otimes \frac{\Tr_{\textsc{s}}[(V_{n}^{\dag}V_{n}\otimes\identity)\rhoSB^{\mathrm{eq}}]}{\Tr[(V_{n}^{\dag}V_{n}\otimes\identity)\rhoSB^{\mathrm{eq}}]}\\
	& = \sum_{ij\alpha n}Q_{ij\alpha n}\rhoS Q_{ij\alpha n}^{\dag}
\end{align}
\end{subequations}
where $Q_{ij\alpha n}\in\BL(\HS;\HSB)$ is the operator
\begin{equation}
	Q_{ij\alpha n} = V_{n}\otimes\frac{\braket{i}{(V_{n}\otimes\identity)\sqrt{\rhoSB^{\mathrm{eq}}}|j\alpha}}{\sqrt{\Tr[(V_{n}^{\dag}V_{n}\otimes\identity)\rhoSB^{\mathrm{eq}}]}}
\end{equation}
so that $\AM_{\CS}$ is CP.  Since $\TrB\circ\AM_{\CS} = T$, the correct domain of definition of this assignment map is the $\mathbb{C}$-linear space of fixed points of $T$ in $\BHS$.  Letting $\CS'$ be the image of this domain through $\AM_{\CS}$, we may define a $\USB$-consistent subspace by $\CS = \Span_{\mathbb{C}}(\DSB\cap\CS')$.  Since the assignment map $\AM_{\CS}$ is CP, it is clear that the dynamical maps $\Psi_{U}^{\CS}:\TrB\CS\to\BHS$ are also CP for all $U\in\USB$.  It is also clear that, although this construction is described by Alicki as involving a loss of consistency, it is entirely valid within the framework of $\Sgp$-consistent subspaces.  It should be noted, however, that this construction depends not only on the map $T$ and the ``equilibrium state'' $\rhoSB^{\mathrm{eq}}$, but also on the \emph{representation} of $T$.  In general, two equivalent OSRs $T = \sum V_{n}\cdot V_{n}^{\dag} = \sum R_{m}\cdot R_{m}^{\dag}$ will not yield the same consistent subspace $\CS$, thus will not yield the same assignment map $\AM_{\CS}$ or dynamical maps $\Psi_{U}^{\CS}$. 

Alicki goes on to consider a generally nonlinear assignment map 
\begin{equation}
	\AM(\rhoS) = \sum \lambda_{n}P_{n}\otimes \frac{\Tr_{\textsc{s}}[\rhoSB^{\mathrm{eq}}(P_{n}\otimes\identity)]}{\Tr[\rhoSB^{\mathrm{eq}}(P_{n}\otimes\identity)]}
\end{equation} 
where again $\rhoSB^{\mathrm{eq}}\in\DSB$ is a fixed ``equilibrium state'', and $\rhoS = \sum \lambda_{n}P_{n}$ is the spectral decomposition of $\rhoS$.  Then $\TrB\circ\AM = \id$ on all of $\DS$, but $\AM:\DS\to\DSB$ is convex-linear if and only if $\rhoSB^{\mathrm{eq}}$ is a tensor product state.  In either case, the image of this assignment map $\AS = \AM(\DS)\subset\DSB$ is a valid $\USB$-consistent subset as described in Section \ref{sec:subsysDynMaps}, contained within the subset of zero-discord states in $\DSB$ \cite{Ollivier:2002} since any state in the image of $\AM$ is of the form $\sum \lambda_{n}P_{n}\otimes\sigma_{n}$ for orthogonal projectors $\{P_{n}\}$ and states $\{\sigma_{n}\}\subset\DSB$. As such, $\rhoSB^{\mathrm{eq}}\in\AS$ only if it has this very special form.
It is clear that the resulting dynamical maps $\tau_{U}^{\AS}$ will be positive, but they will also generally \emph{not} be convex-linear.  
However, the physical meaning of $\rhoSB^{\mathrm{eq}}$ is unclear when $\rhoSB^{\mathrm{eq}}\notin\AS$, and moreover $\tau_{U}^{\AS}\otimes\id$ and $\AM\otimes\id$ have no meaning when $\tau_{U}^{\AS}$ and $\AM$ are nonlinear, so that complete positivity has no meaning either.

\subsection{\v{S}telmachovi\texorpdfstring{\v{c}}{c} and Bu\v{z}ek (2001)}

In addition to Example~\ref{ex:StelBu}, \v{S}telmachovi\v{c} and Bu\v{z}ek \cite{Stelmachovic:2001} offer a second example which demonstrates that, if we \emph{fix} any $\rhoS(0),\rhoS(T)\in\DS$, then it is possible to choose $\HB$, a $\USB$-consistent subspace $\CS\subset\BHSB$, and a unitary transformation $U\in\USB$ such that the corresponding dynamical map $\Psi_{U}^{\CS}$ takes $\rhoS(0)$ to $\rhoS(T)$.  Indeed, they show this can be done with a Kraus-like consistent subspace $\CS = \BHS\otimes\rhoB$, where $\HB\simeq \HS$ and  $\rhoB = \rhoS(T)$, and where the unitary transformation $U = \sum\ketbra{i}{j}\otimes\ketbra{j}{i}$ is the swap operator $U\ket{\psi}\otimes\ket{\phi} = \ket{\phi}\otimes\ket{\psi}$.

\subsection{Jordan, Shaji, and Sudarshan (2004)}
Jordan, Shaji, and Sudarshan \cite{Jordan:2004} develop a two qubit example (one ``system'' qubit and one ``bath'' qubit) in detail in order to examine quantum dynamics in the presence of initial entanglement.  Their example can be expressed in terms of a $\Sgp$-consistent subspace, where $\Sgp\subset \USB$ is the one-parameter subsemigroup generated by the Hamiltonian $H = (\omega/2) Z\otimes X$ where $Z$ and $X$ are Pauli operators.  The $\Sgp$-consistent subspace $\CS\subset \BHSB$ is then the 14-dimensional orthogonal complement (with respect to the Hilbert-Schmidt inner product) of $\Span_{\mathbb{C}}\{\alpha\identity+ Y\otimes X, \beta\identity-X\otimes X\}$ for some fixed choice of $\alpha,\beta\in(-1,1)$.  The kernel $\CS_{0} = \ker\big(\TrB\big|_{\CS}\big)$ is then the 10-dimensional subspace
\begin{align}
	\CS_{0} & = \Span_{\mathbb{C}}\{\identity\otimes X,\identity\otimes Y,\identity\otimes Z, X\otimes Y, X\otimes Z, Y\otimes Y,\nonumber\\
	 & \qquad \qquad \quad  Y\otimes Z, Z\otimes X, Z\otimes Y, Z\otimes Z\}.
\end{align}
It is straightforward to check that $\Sgp\cdot\CS_{0}\subset\ker\TrB$ as required.  The quotient space $\CS/\CS_{0}$ is then 4-dimensional, so that $\TrB\CS = \BHS$.  However, as pointed out in \cite{Jordan:2004}, unless $\alpha = \beta = 0$, the assignment map $\AM_{\CS}:\TrB\CS\to\CS/\CS_{0}$ is not positive.  If it were positive, then $\AM_{\CS}(\ketbra{0}{0})$ would be an affine subspace of $\BHSB$ containing a state of the form $\ketbra{0}{0}\otimes\rhoB$ for some $\rhoB\in\DB$.  But unless $\alpha = \beta = 0$, no state of that form lies in $\CS$ since such states are not orthogonal to $\alpha\identity+ Y\otimes X$ and $\beta\identity-X\otimes X$.  It follows that the affine subspace $\AM_{\CS}(\ketbra{0}{0})$ contains only unit trace, non-positive operators, so that $\AM_{\CS}$ is non-positive. Moreover it is shown that the $\Psi_{U}^{\CS}(\ketbra{0}{0})$ need not be positive for $U\in \Sgp$ and therefore $\Psi_{U}^{\CS}$ is not completely positive for all $U\in \Sgp$.  On the other hand, if $\alpha = \beta = 0$, then $\CS$ contains the $\USB$-consistent subspace $\hat{\CS} = \BHS\otimes\identity$ for which $\CS = \hat{\CS}\oplus\CS_{0}$.  In that case, the assignment map is not only positive, but completely positive.

\subsection{Carteret, Terno, and \.{Z}yczkowski (2008)}
Carteret \emph{et al.} \cite{Carteret:2008}, consider an example with a one qubit system and bath and a $\USB$-consistent subspace of the form 
\begin{equation}
	\CS = \Span_{\mathbb{C}}\left\{\sigma_{1}\otimes\identity, \sigma_{2}\otimes\identity, \sigma_{3}\otimes\identity, \identity+a\sum_{i=1}^{3}\sigma_{i}\otimes\sigma_{i}\right\},
\end{equation}
where $\{\sigma_{i}\}$ are the Pauli operators and $-1<a<1/3$ is a fixed constant.  For $a \neq 0$, the associated assignment map $\AM_{\CS}$ is non-positive, admitting a physical domain $\TrB(\DSB\cap\CS)$ which, within the Bloch sphere, is the ball centered at $\identity/2$ with radius $\sqrt{(1+a)(1-3a)}$ for $a\geq 0$ and with radius $1+a$ for $a\leq 0$ (this latter point about $a<0$ was missed in \cite{Carteret:2008}). The authors note that the dynamical map $\Psi_{U}^{\CS}$ is trivially completely positive for any $U\in\USB$ that commutes with $\sum\sigma_{i}\otimes\sigma_{i}$.  They then consider unitary transformations of the form
\begin{equation}
	U = \begin{pmatrix}1 & 0 & 0 & 0 \\ 0 & \cos \theta & \sin \theta  & 0 \\ 0 & -\sin \theta & \cos \theta & 0 \\ 0 & 0 & 0 & 1\end{pmatrix},
\end{equation}
pointing out that the $\theta = \pi/4$ case yields a CP dynamical map $\Psi_{U}^{\CS}$ and incorrectly claiming that $\theta = \pi$ yields a non-CP dynamical map (this case also yields a CP $\Psi_{U}^{\CS}$, 
because for $\theta = \pi$, $U=\sigma_{z}\otimes\sigma_{z}$ commutes with $\sum \sigma_{i}\otimes\sigma_{i}$).  However, it may be shown that the Choi matrix for the dynamical map $\Psi_{U}^{\CS}$ associated with such a $U$ is given by
\begin{align}
	\mathrm{Choi}(\Psi_{U}^{\CS}) & = \frac{1}{2}\big[\identity + \cos^{2}(\theta)\sigma_{z}\otimes\sigma_{z} + a\sin(2\theta)\identity\otimes\sigma_{z} \nonumber\\
	& \qquad + \cos(\theta)\big(\sigma_{x}\otimes\sigma_{x} - \sigma_{y}\otimes\sigma_{y}\big)\big]
\end{align}
so that other values for $\theta$ and $a$ can yield non-CP maps, for example the case $\theta = \pi/6$ and $a > 1/2\sqrt{3}$.

\subsection{Rodr\'{\i}guez-Rosario, Modi, Kuah, Shaji, and Sudarshan (2008)}
Rodr\'{\i}guez-Rosario \emph{et al.} \cite{Rod:08}, considered $\USB$-consistent subspaces of the form 
\begin{equation}
	\label{eq:RR08CS}
	\CS = \Span_{\mathbb{C}}\{\ketbra{i}{i}\otimes\sigma_{i}\},
\end{equation}
where $\{\ket{i}\}\subset\HS$ is an orthonormal basis for the system and $\{\sigma_{i}\}\subset\DB$ are bath states.  Then $\CS\cap\DSB$ comprises only zero-discord states, i.e., those that exhibit only classical correlations with respect to some measurement basis \cite{Ollivier:2002}.  They showed that, for such a $\CS$ and any $U\in\USB$, the corresponding dynamical map $\Psi_{U}^{\CS}$ is always CP. In fact, more is true: the assignment map $\AM_{\CS}:\TrB\CS\to\CS$ is CPZE since $\AM_{\CS}\circ\mathcal{P}_{\TrB\CS}(X) = \sum_{i,\alpha}E_{i\alpha}X E_{i,\alpha}^{\dag}$ is CP, where $\mathcal{P}_{\TrB\CS}:\BHS\to\TrB\CS$ is the orthogonal projection onto $\TrB\CS$ and $E_{i,\alpha} = \ketbra{i}{i}\otimes\sqrt{\sigma_{i}}\ket{\alpha}$ for some orthonormal basis $\{\ket{\alpha}\}\subset\HB$.  

As we discussed in Section \ref{sec:constrainedDomains}, whenever the semigroup $\Sgp$ contains nonlocal unitary operators, all $\Sgp$-consistent subspaces must be constrained in some way.  Moreover, if we accept that such constraints on the initial system-bath states are inevitable, there seems to be little reason to disallow constraints on $\TrB\CS$, which is the domain of the assignment map $\AM_{\CS}$ and of the dynamical maps $\Psi_{U}^{\CS}$ for all $U\in\Sgp$.  The authors of \cite{Rod:08} seem to take another position, however, stating in the introduction of that paper that ``the dynamical map is well-defined if it is positive on a large enough set of states such that it can be extended by linearity to all states of the system.''  It may be observed, as the authors themselves do, that for the $\dim(\HS)$-dimensional consistent subspaces $\CS$ in \eqref{eq:RR08CS}, the linear domain of definition, $\TrB\CS$, of the dynamical maps consists only of operators which are diagonal in the fixed basis $\{\ket{i}\}$.  Such dynamical maps are, in fact, not ill-defined at all, contrary to the position adopted in \cite{Rod:08}.  They are well-defined maps with well-defined domains which are low-dimensional subspaces of $\BHS$.  In addition, they are associated with a well-defined, CP assignment map $\AM_{\CS}:\TrB\CS\to\CS$, so that the complete positivity of the dynamical maps has a good, physical basis related to the unitary evolution of system-bath-witness states for arbitrary finite-dimensional witnesses (see Section \ref{sec:whyCP}).

\subsection{Shabani and Lidar (2009)}
\label{sec:ShabaniLidar2009}
Shabani and Lidar \cite{SL} expanded the class of consistent subspaces from the zero-discord subspaces considered in \cite{Rod:08}, to the class of all valid $\USB$-consistent subspaces of the form 
\begin{align}
	\CS = \Span_{\mathbb{C}}\{\ketbra{i}{j}\otimes \phi_{ij}\},
	\label{eq:SL}
\end{align}
where $\{\ket{i}\}$ is an orthonormal basis for $\HS$ and $\{\phi_{ij}\}\subset\BHB$ are bath operators.  Ref.~\cite{SL} showed that, within this class of consistent subspaces, the dynamical maps $\Psi_{U}^{\CS}$ are CP for all $U\in\USB$ if and only if $\CS\cap\DSB$ comprises only zero-discord states.  

In light of the framework developed here, it is clear that the results of \cite{SL} do \emph{not} amount in general to necessary and sufficient conditions for complete positivity, but are restricted to the class of consistent subspaces of the form \eqref{eq:SL}.  This lack of generality is a point that has been largely overlooked by the community, including by two of the present authors. Thus the present work amounts to a retraction and correction of some of the claims made in \cite{SL}. The issue is that, while the form \eqref{eq:SL} is general enough to describe any given state $\rhoSB$, it is not general enough to describe every $\Sgp$-consistent subspace.  As a simple example, the $2$-dimensional $\USB$-consistent subspace described in Example~\ref{ex:consistentPosCounterEx} is not of the form \eqref{eq:SL}, nor is it contained in any subspace of this form.

Let us also comment on \cite{SL2}, which does not address the question of whether the maps it constructs are CP or not, but states in its Theorem 3, that ``the most general form of a quantum dynamical process irrespective of the initial system-bath state (in particular arbitrarily entangled initial states are possible) is always reducible to a Hermitian map from the initial system to the final system state.'' We note that  this only holds for $\Sgp$-consistent subspaces; when the set of admissible initial states $\AS$ is not  $\Sgp$-consistent, there is no subsystem $\AS$-dynamical map (for some $U\in\Sgp$).

\subsection{Brodutch, Datta,  Modi, Rivas, and Rodr\'iguez-Rosario (2013)}
To demonstrate that vanishing discord is not necessary for complete positivity, Brodutch \textit{et al.} \cite{Brodutch:2013} showed that $\USB$-consistent subspaces $\CS$ exist for which almost all states in $\CS\cap\DSB$ are discordant, but the dynamical maps $\Psi_{U}^{\CS}$ are nonetheless CPZE (and therefore CP) for all $U\in\USB$.  The counterexamples described in \cite{Brodutch:2013} are subspaces of the form 
\begin{equation}
	\CS = \Span_{\mathbb{C}}\{\rho_{01}\}\cup\{\ketbra{i}{i}\otimes\sigma_{i}\}_{i=2}^{n}
\end{equation}
where $\rho_{01} = \ketbra{0}{0}\otimes \sigma_{0} + \ketbra{1}{1}\otimes\sigma_{1} + \ketbra{+}{+}\otimes\sigma_{+}$, $\{\ket{i}\}$ is an orthornomal basis for $\HS$, and $\{\sigma_{+}, \sigma_{0}, \dots, \sigma_{n}\}\subset\DB$ are bath states.  

It should be mentioned that, while Brodutch, \textit{et al.} were correct to criticize the generality of the conclusion concerning the necessity of vanishing discord for complete positivity in \cite{SL}, their claim of nonlinearity appearing in \cite{SL} is not valid.  By expressing the construction of \cite{SL} in the language of $\Sgp$-consistent subspaces as in \eqref{eq:SL} it should be clear that, while the initial system-bath states lie in a constrained subspace -- as they always must for $\Sgp$-consistency, i.e., to obtain well-defined dynamical maps (see Lemma \ref{lem:localUconsistency})  -- all spaces and maps in the Shabani-Lidar framework are $\mathbb{C}$-linear.

\subsection{Buscemi (2013)}
Buscemi \cite{Buscemi:2014} has devised a general method for constructing $\USB$-consistent subspaces $\CS$ which yield CP dynamical maps $\Psi_{U}^{\CS}$ for all $U\in\USB$.  The technique involves first introducing a ``reference'' subsystem $R$ and choosing a fixed tripartite state $\rho_{\textsc{rsb}}$ that saturates the strong subadditivity inequality for von Neumann entropy; in other words the conditional mutual information, $I(R:B|S)_{\rho}$, is zero.  A subspace $\CS'\subset\BHSB$ is then formed as
\begin{equation}
	\CS' = \{\Tr_{\textsc{r}}[(L\otimes\identity_{\textsc{sb}})\rho_{\textsc{rsb}}]\;:\; L\in\BL(\HH_{\textsc{r}})\}
\end{equation}
and the corresponding $\USB$-consistent subspace is $\CS = \Span_{\mathbb{C}}(\DSB\cap\CS')$.  Buscemi proves that such a construction always leads to CP dynamical maps.  While there is no claim that all $\USB$-consistent subspaces yielding CP dynamical maps must be formed in this way, the technique does yield a rich set of examples, including the zero-discord subspaces of Rodr\'{\i}guez-Rosario \emph{et al.} \cite{Rod:08} and the examples of Brodutch \emph{et al.} \cite{Brodutch:2013}, as well as examples featuring highly entangled states.  It should also be noted that, although Buscemi's method predicts that the dynamical maps will be CP in these examples, it does not appear to predict the fact that the assignment maps are CP (indeed, CPZE) in the examples of \cite{Rod:08} and \cite{Brodutch:2013}.

\section{Summary \& Open Questions}
\label{sec:Summary}

We have formulated a general framework for  open quantum system dynamics in the presence of initial correlations between system and bath.  It is based around the notion of $\Sgp$-consistent operator spaces $\CS\subset\BHSB$ representing the set of admissible initial system-bath states.  $\Sgp$-consistency is necessary to ensure that initial system-bath states with the same reduced state on the system evolve under all admissible unitary operators $U\in\Sgp$ to system-bath states with the same reduced state on the system, ensuring the dynamical maps $\Psi_{U}^{\CS}$ are well-defined.  Once such a $\Sgp$-consistent subspace is chosen, related concepts like the assignment map and the dynamical maps are uniquely defined.  In general, the dynamical maps may not be applied to arbitrary system states, but only to those in the physical domain $\TrB(\DSB\cap\CS)$. Interestingly, non-positive assignment maps can give rise to completely positive dynamical maps, but this is not a situation that is physically acceptable since it involves the propagation of unphysical system-bath ``states".

These $\Sgp$-consistent subspaces may be highly constrained, with the result that complete positivity of the assignment and dynamical maps may be defined in several distinct ways, relating to different interpretations and desirable features of complete positivity for open system dynamics. Indeed, we have identified two new types of completely positive maps, which we have called CPTE and CPZE. It is important to distinguish between these flavors of complete positivity and whether one is considering complete positivity of the assignment map or of the dynamical maps. {We have shown that various special cases considered in earlier work can be unified within the framework of $\Sgp$-consistency, and that earlier discussions of complete positivity are better understood using CP, CPTE, or CPZE maps.}

Having laid out this framework of $\Sgp$-consistent subspaces, we have left many important questions open.  The formalism is general enough to encompass other proposed frameworks, most notably the assignment map framework introduced by Pechukas \cite{Pechukas:1994}, and various worked examples.  These various results and special cases illustrate perhaps the biggest open question: what structural features characterize the subclass of $\Sgp$-consistent subspaces which give rise to completely positive dynamics (in any reasonable flavor)? Recent work \cite{Rod:08,SL} 
suggested that the answer might lie with quantum discord. However, more recently it became apparent that the focus on quantum discord is too restrictive \cite{Brodutch:2013,Buscemi:2014}, and here we have shown that earlier work has  illuminated some small corners of the space of $\Sgp$-consistent subspaces, but a complete analysis remains to be seen.

\acknowledgments
This research was supported by the ARO MURI grant W911NF-11-1-0268 and by NSF grant numbers PHY-969969 and PHY-803304.  

\bibliography{refs}

\end{document}